\documentclass[12pt,a4paper]{amsart}

\usepackage{amsmath,amsfonts,amssymb}
\usepackage[hmargin=2cm,vmargin=2cm]{geometry}
\usepackage{hyperref}
\usepackage{nameref,zref-xr}                    
\usepackage{amsthm}

\newcommand{\CC}{{\mathbb{C}}}

\def\T{{\mathcal T}}

\newcommand{\p}{{\partial}}
\newcommand{\mbC}{\mathbb C}
\newcommand{\oM}{\overline{\mathcal M}}
\def\mbQ{{\mathbb Q}}
\def\d{{\partial}}
\newcommand{\<}{\left<}
\renewcommand{\>}{\right>}

\newcommand{\Coef}{\mathrm{Coef}}

\newcommand{\mcL}{\mathcal{L}}
\newcommand{\mcO}{\mathcal{O}}
\newcommand{\tD}{\widetilde{D}}
\newcommand{\talpha}{{\widetilde{\alpha}}}
\newcommand{\tbeta}{{\widetilde{\beta}}}
\newcommand{\ext}{\mathrm{ext}}
\newcommand{\mcA}{\mathcal{A}}

\newcommand{\mcT}{\mathcal{T}}
\newcommand{\alg}{\mathrm{alg}}

\newcommand{\rspin}{\text{$r$-spin}}
\newcommand{\oalpha}{{\overline{\alpha}}}
\newcommand{\hmcA}{{\widehat{\mathcal{A}}}}

\newtheorem{theorem}{Theorem}[section]

\newtheorem{lemma}[theorem]{Lemma}
\newtheorem{corollary}[theorem]{Corollary}

\newtheorem{definition}[theorem]{Definition}
\newtheorem{example}[theorem]{Example}
\newtheorem{remark}[theorem]{Remark}

\numberwithin{equation}{section}

\begin{document}

\title{Open Saito theory for $A$ and $D$ singularities}

\author{Alexey Basalaev}
\address{A. Basalaev:\newline Faculty of Mathematics, National Research University Higher School of Economics, Usacheva str., 6, 119048 Moscow, Russian Federation, and \newline
Skolkovo Institute of Science and Technology, Nobelya str., 3, 121205 Moscow, Russian Federation}
\email{a.basalaev@skoltech.ru}

\author{Alexandr Buryak}
\address{A. Buryak:\newline School of Mathematics, University of Leeds, Leeds, LS2 9JT, United Kingdom}
\email{a.buryak@leeds.ac.uk}

\date{\today}

\begin{abstract}
A well known construction of B.~Dubrovin and K.~Saito endows the parameter space of a universal unfolding of a simple singularity with a Frobenius manifold structure. In our paper we present a generalization of this construction for the singularities of types~$A$ and~$D$, that gives a solution of the open WDVV equations. For the $A$-singularity the resulting solution describes the intersection numbers on the moduli space of $r$-spin disks, introduced recently in a work of the second author, E.~Clader and R.~Tessler. In the second part of the paper we describe the space of homogeneous polynomial solutions of the open WDVV equations associated to the Frobenius manifolds of finite irreducible Coxeter groups.  
\end{abstract}

\maketitle

\setcounter{tocdepth}{1}
\tableofcontents

\section{Introduction}

Frobenius manifolds, introduced by B.~Dubrovin in the early 90s, gave a geometric approach to study solutions of the {\it WDVV equations}
\begin{gather}\label{eq:WDVV equations}
\frac{\d^3 F}{\d t^\alpha\d t^\beta\d t^\mu}\eta^{\mu\nu}\frac{\d^3 F}{\d t^\nu\d t^\gamma\d t^\delta}=\frac{\d^3 F}{\d t^\delta\d t^\beta\d t^\mu}\eta^{\mu\nu}\frac{\d^3 F}{\d t^\nu\d t^\gamma\d t^\alpha},\quad 1\le\alpha,\beta,\gamma,\delta\le N,
\end{gather}
where $F = F(t^1,\dots,t^N)$ is an analytic function defined on some open subset $M\subset\mbC^N$, $\eta=(\eta_{\alpha\beta})$ is an $N\times N$ symmetric non-degenerate matrix with complex coefficients, $(\eta^{\alpha\beta}):=\eta^{-1}$ and we use the convention of sum over repeated Greek indices. The WDVV equations appear in many areas of mathematics, including singularity theory and curve counting theories in algebraic geometry. In Gromov--Witten theory the WDVV equations describe the structure of primary Gromov--Witten invariants in genus $0$ and naturally come from a certain relation in the cohomology of the moduli space of stable curves.

\bigskip

Suppose that a function $F$ satisfies the WDVV equations together with the additional assumption 
\begin{gather}\label{eq:unit condition for F}
\frac{\d^3 F}{\d t^1\d t^\alpha\d t^\beta}=\eta_{\alpha\beta}.
\end{gather}
The function $F$ defines a commutative product $\circ$ on each tangent space $T_p M$ by
\[
  \frac{\p}{\p t^\alpha} \circ \frac{\p}{\p t^\beta}:=\frac{\p^3F}{\p t^\alpha \p t^\beta \p t^\gamma} \eta^{\gamma \delta} \frac{\p}{\p t^\delta}, \quad 1 \le \alpha,\beta \le N.
\]
One can immediately see that the WDVV equations are equivalent to the associativity of this product and property~\eqref{eq:unit condition for F} means that the vector field $\frac{\d}{\d t^1}$ is the unit. One can go in the opposite direction and consider a manifold with a commutative, associative algebra structure and a symmetric, non-degenerate bilinear form on each tangent space. Under  certain conditions such a manifold in special local coordinates, called the {\it flat coordinates}, can be described by a solution~$F$ of the WDVV equations, satisfying property~\eqref{eq:unit condition for F}. The conditions, needed for the existence of a function~$F$, were systematically studied by B.~Dubrovin~\cite{Dub96,Dub99}, who called manifolds, satisfying these conditions, {\it Frobenius manifolds}. The function $F$ is then called a {\it Frobenius manifold potential}. The bilinear form is traditionally called a {\it metric}.

\bigskip

In his fundamental works~\cite{Sai82,Sai83} K.~Saito constructed a flat metric on the parameter space of a universal unfolding of any simple singularity. B.~Dubrovin~\cite{Dub98} then proved that together with a certain commutative, associative algebra structure on each tangent space this metric defines a Frobenius manifold structure on the parameter space of the universal unfolding. These Frobenius manifolds are often called the {\it Saito Frobenius manifolds}. Remarkably, the same Frobenius manifolds appear in the study of the geometry of the moduli spaces of algebraic curves with certain additional structures, the so-called Fan--Jarvis--Ruan--Witten (FJRW) theory~\cite{FJR13}. This is one of the manifestations of mirror symmetry. 

\bigskip 

In the same way, as the WDVV equations appeared in Gromov--Witten theory, another system of non-linear PDEs, called the \textit{open WDVV equations}, appeared more recently in open Gromov--Witten theory~\cite[Theorem~2.7]{HS12} (see also~\cite{PST14,BCT18,BCT19}). Let $F = F(t^1,\dots,t^N)$ be a solution of the WDVV equations~\eqref{eq:WDVV equations}, satisfying condition~\eqref{eq:unit condition for F}. The open WDVV equations associated to $F$ are the following PDEs for a function $F^o = F^o(t^1,\ldots,t^N,s)$, depending on an additional variable $s$:
\begin{align}
\frac{\d^3F}{\d t^\alpha\d t^\beta\d t^\mu}\eta^{\mu\nu}\frac{\d^2F^o}{\d t^\nu\d t^\gamma}+\frac{\d^2F^o}{\d t^\alpha\d t^\beta}\frac{\d^2F^o}{\d s\d t^\gamma}=&\frac{\d^3F}{\d t^\gamma\d t^\beta\d t^\mu}\eta^{\mu\nu}\frac{\d^2F^o}{\d t^\nu\d t^\alpha}+\frac{\d^2F^o}{\d t^\gamma\d t^\beta}\frac{\d^2F^o}{\d s\d t^\alpha},&&1\le\alpha,\beta,\gamma\le N,\label{eq:open WDVV,1}\\
\frac{\d^3F}{\d t^\alpha\d t^\beta\d t^\mu}\eta^{\mu\nu}\frac{\d^2F^o}{\d t^\nu\d s}+\frac{\d^2F^o}{\d t^\alpha\d t^\beta}\frac{\d^2F^o}{\d s^2}=&\frac{\d^2F^o}{\d s\d t^\beta}\frac{\d^2F^o}{\d s\d t^\alpha},&&1\le\alpha,\beta\le N.\label{eq:open WDVV,2}
\end{align}
Solutions of equations~\eqref{eq:open WDVV,1},~\eqref{eq:open WDVV,2}, relevant in open Gromov-Witten theory and also in the works~\cite{PST14,BCT18,BCT19}, satisfy the additional condition
\begin{gather}\label{eq:unit condition for Fo}
\frac{\d^2 F^o}{\d t^1\d t^\alpha}=0,\qquad \frac{\d^2 F^o}{\d t^1\d s}=1.
\end{gather}

\bigskip

The solutions of the open WDVV equations from the works~\cite{BCT18,BCT19} are associated to the Saito Frobenius manifold of the $A$-singularity and they were constructed using ideas of FJRW theory. So it is natural to ask whether the Dubrovin--Saito construction of the Frobenius manifolds corresponding to simple singularities admits a generalization, that produces solutions of the open WDVV equations. In our paper we present such a generalization for the singularities of types $A$ and $D$. For the $A$-singularity our construction gives a polynomial solution that coincides with the one from~\cite{BCT18,BCT19}. For the $D$-singularity our solution has a simple pole along the variable $s$. 

\bigskip

Additionally, in both $A$- and $D$-cases our solution of the open WDVV equations has the following remarkable feature. The Saito Frobenius manifold of a simple singularity has two natural coordinate systems. The first one is given by the parameters of a universal unfolding of a simple singularity. The second coordinate system is given by the flat coordinates of the metric. We show that for the singularities $A$ and $D$ the transition functions between these two coordinate systems coincide with the coefficients of powers of the variable~$s$ in the expansion of our solution of the open WDVV equations.    

\bigskip

The Saito Frobenius manifolds of simple singularities together with their certain submanifolds form a class of Frobenius manifolds, that is, via a construction of B.~Dubrovin~\cite{Dub98}, in a natural bijection with the class of finite irreducible Coxeter groups (see also~\cite{Zub94}). This class of Frobenius manifolds plays a fundamental role in the theory of Frobenius manifolds, because of the following result of C.~Hertling, conjectured by B.~Dubrovin~\cite{Dub98}. Recall that a Frobenius manifold potential $F$ is called homogeneous, if there exists a vector field $E$ of the form
\begin{gather}\label{eq:Euler field for Frobenius}
E=\sum_{\alpha=1}^N\underbrace{(q_\alpha t^\alpha+r^\alpha)}_{=:E^\alpha}\frac{\d}{\d t^\alpha},\quad q_\alpha,r^\alpha\in\mbC,\quad q_1=1,
\end{gather}
such that
\begin{gather*}
E(F)=E^\alpha\frac{\d F}{\d t^\alpha}=(3-\delta)F+\frac{1}{2}A_{\alpha\beta}t^\alpha t^\beta+B_\alpha t^\alpha+C,\quad\text{for some $\delta,A_{\alpha\beta},B_\alpha,C\in\mbC$}.
\end{gather*}
The number $\delta$ is called the {\it conformal dimension} and the vector field $E$ is called the {\it Euler vector field}. C.~Hertling proved that any generically semisimple Frobenius manifold (see Section~\ref{subsection:flat F-manifolds} for definition), whose potential is polynomial $F \in \CC[t^1,\dots,t^N]$ and homogeneous with the Euler vector field of the form~$E=q_\alpha t^\alpha\frac{\d}{\d t^\alpha}$, where $q_\alpha>0$, can be expressed as the product of the Frobenius manifolds corresponding to finite irreducible Coxeter groups~\cite[Theorem~5.25]{Hert02}.

\bigskip

In the second part of the paper we study the space of polynomial solutions of the open WDVV equations associated to the Frobenius manifolds of finite irreducible Coxeter groups. Note that all solutions of the open WDVV equations, considered in the works~\cite{HS12,PST14,BCT18,BCT19}, are associated to a homogeneous Frobenius potential~$F$ and, moreover, the function $F^o$ satisfies the homogeneity condition
\begin{gather}\label{eq:homogeneity for Fo}
E^\alpha\frac{\d F^o}{\d t^\alpha}+\frac{1-\delta}{2}s\frac{\d F^o}{\d s}=\frac{3-\delta}{2}F^o+D_\alpha t^\alpha+\tD s+E,\quad\text{for some $D_\alpha,\tD,E\in\mbC$}.
\end{gather}
We see that the degree of the variable~$s$ is determined by the conformal dimension of the Frobenius manifold. We will call a solution of the open WDVV equations homogeneous, if it satisfies condition~\eqref{eq:homogeneity for Fo}.

\bigskip 

In our paper we describe the space of homogeneous polynomial solutions of the open WDVV equations associated to the Frobenius manifolds of {\it all} finite irreducible Coxeter groups. In particular, this space is non-empty only for the Coxeter groups~$A_N$,~$B_N$ and~$I_2(k)$. 

\bigskip

Our approach to study solutions of the open WDVV equations is based on the following crucial observation of P. Rossi. Let $F = F(t^1,\ldots,t^N)$ be a Frobenius manifold potential and $F^o=F^o(t^1,\ldots,t^N,s)$ be a solution of the open WDVV equations satisfying~\eqref{eq:unit condition for Fo}. Then the $(N+1)$-tuple of functions~$\left(\eta^{1\mu}\frac{\d F}{\d t^\mu},\ldots,\eta^{N\mu}\frac{\d F}{\d t^\mu},F^o\right)$ forms a {\it vector potential} of a {\it flat F-manifold}. This allows us to use the theory of flat F-manifolds to study solutions of the open WDVV equations.

\subsection*{Acknowledgements}

We would like to thank Claus Hertling for useful discussions. 

This project has received funding from the European Union's Horizon 2020 research and innovation programme under the Marie Sk\l odowska-Curie grant agreement No. 797635. The first named author was supported by RSF Grant No. 19-71-00086.


\section{Flat F-manifolds and Frobenius manifolds}\label{section:definitions}

In this section we recall the definitions and the main properties of flat F-manifolds and Frobenius manifolds. We also explain how solutions of the open WDVV equations correspond to flat F-manifolds of special type.

\subsection{Flat F-manifolds}\label{subsection:flat F-manifolds}

The notion of a flat F-manifold was introduced in~\cite{Man05}. 

\begin{definition}
A flat F-manifold $(M,\nabla,\circ)$ is the datum of a complex analytic manifold $M$, an analytic connection $\nabla$ in the tangent bundle $T M$, an algebra structure $(T_p M,\circ)$ with unit $e$ on each tangent space analytically depending on the point $p\in M$ such that the one-parameter family of connections $\nabla+z\circ$ is flat and torsionless for any $z\in\mbC$, and $\nabla e=0$.
\end{definition}

For a flat F-manifold $(M,\nabla,\circ)$ consider flat coordinates $t^\alpha$, $1\le\alpha\le N$, $N=\dim M$, for the connection $\nabla$ such that $e = \frac{\d}{\d t^1}$. Then locally there exist analytic functions $F^\alpha(t^1,\ldots,t^N)$, $1\leq\alpha\leq N$, such that the second derivatives 
\begin{gather}\label{eq:structure constants of flat F-man}
c^\alpha_{\beta\gamma}=\frac{\d^2 F^\alpha}{\d t^\beta \d t^\gamma}
\end{gather}
give the structure constants for the multiplication $\circ$,
\begin{gather*}
\frac{\d}{\d t^\beta}\circ\frac{\d}{\d t^\gamma}=c^\alpha_{\beta\gamma}\frac{\d}{\d t^\alpha}.
\end{gather*}
From the associativity of the multiplication and the fact that the vector field $\frac{\d}{\d t^1}$ is the unit it follows that
\begin{align}
\frac{\d^2 F^\alpha}{\d t^1\d t^\beta} &= \delta^\alpha_\beta, && 1\leq \alpha,\beta\leq N,\label{eq:axiom1 of flat F-man}\\
\frac{\d^2 F^\alpha}{\d t^\beta \d t^\mu} \frac{\d^2 F^\mu}{\d t^\gamma \d t^\delta} &= \frac{\d^2 F^\alpha}{\d t^\gamma \d t^\mu} \frac{\d^2 F^\mu}{\d t^\beta \d t^\delta}, && 1\leq \alpha,\beta,\gamma,\delta\leq N.\label{eq:axiom2 of flat F-man}
\end{align}
The $N$-tuple of functions $(F^1,\ldots,F^N)$ is called the {\it vector potential} of our flat F-manifold.

Conversely, if $M$ is an open subset of $\mbC^N$ and $F^1,\ldots,F^N\in\mcO(M)$ are functions satisfying equations~\eqref{eq:axiom1 of flat F-man} and~\eqref{eq:axiom2 of flat F-man}, then these functions define a flat F-manifold $(M,\nabla,\circ)$ with the connection~$\nabla$, given by $\nabla_{\frac{\d}{\d t^\alpha}}\frac{\d}{\d t^\beta}=0$, and the multiplication $\circ$, given by the structure constants~\eqref{eq:structure constants of flat F-man}. 

\begin{remark}\label{remark:algebraic construction of multiplication}
Let $M\subset\mbC^N$ be an open subset in the Zariski topology. The tangent spaces $T_p M$ can be endowed with an algebra structure, algebraically depending on the point $p\in M$, using the following construction. Denote by $\mcO^\alg$ the sheaf of algebraic functions on $M$. Let $R$ be an $\mcO^\alg(M)$-algebra, which is free as an $\mcO^\alg(M)$-module with a basis $\phi_1,\ldots,\phi_N\in R$. Denote by $v_1,\ldots,v_N$ the standard coordinates on $\mbC^N$ and by $\mcT_M^\alg$ the algebraic tangent sheaf of $M$. Define an isomorphism of $\mcO^\alg(M)$-modules $\Psi\colon\mcT_M^\alg(M)\to R$ by $\Psi(\frac{\d}{\d v_i}):=\phi_i$. Thus, the sheaf $\mcT_M^\alg$ becomes a sheaf of $\mcO^\alg$-algebras, that endows the tangent spaces $T_p M$ with an algebra structure algebraically depending on the point $p\in M$.  
\end{remark}

\begin{remark}
Consider an analytic manifold $M$ with an algebra structure $(T_p M,\circ)$ on each tangent space analytically depending on the point $p\in M$. We see that a connection $\nabla$, endowing our manifold $M$ with a flat F-manifold structure, can be completely described by a choice of coordinates $t^1,\ldots,t^N$ on $M$ such that the structure constants $c^\alpha_{\beta\gamma}$ of multiplication in these coordinates satisfy the integrability condition 
$$
\frac{\d c^\alpha_{\beta\gamma}}{\d t^\sigma}=\frac{\d c^\alpha_{\beta\sigma}}{\d t^\gamma}
$$
together with the condition $c^\alpha_{1,\beta}=\delta^\alpha_\beta$. In this paper we will construct flat $F$-manifolds exactly by presenting flat coordinates as above.
\end{remark}

A flat F-manifold $(M,\nabla,\circ)$ is called {\it conformal}, if it is equipped with a vector field $E$, called the {\it Euler vector field}, such that $\nabla\nabla E=0$, $[e,E]=e$ and $\mcL_E(\circ)=\circ$. This means that in the flat coordinates the Euler vector field $E$ has the form
$$
E=\underbrace{(q^\alpha_\beta t^\beta+r^\alpha)}_{=:E^\alpha}\frac{\d}{\d t^\alpha},\quad q^\alpha_\beta,r^\alpha\in\mbC,\quad q^\alpha_1=\delta^\alpha_1,
$$
and the vector potential $(F^1,\ldots,F^N)$ satisfies the condition
\begin{gather}\label{eq:homogeneity for F-man}
E^\mu\frac{\d F^\alpha}{\d t^\mu}=q^\alpha_\beta F^\beta+F^\alpha+A^\alpha_\beta t^\beta+B^\alpha,\quad\text{for some $A^\alpha_\beta,B^\alpha\in\mbC$}.
\end{gather}

A point $p\in M$ of an $N$-dimensional flat F-manifold $(M,\nabla,\circ)$ is called \textit{semisimple} if $T_pM$ has a basis of idempotents $\pi_1,\dots,\pi_N$ satisfying $\pi_k \circ \pi_l = \delta_{k,l} \pi_k$. Moreover, locally around such a point one can choose coordinates $u^i$ such that $\frac{\d}{\d u^k}\circ\frac{\d}{\d u^l}=\delta_{k,l}\frac{\d}{\d u^k}$. These coordinates are called the {\it canonical coordinates}. In particular, this means that the semisimplicity is an open property. The flat F-manifold $(M,\nabla,\circ)$ is called semisimple, if a generic point of $M$ is semisimple.

\subsection{Frobenius manifolds}

For a complex analytic manifold $M$ we denote by $\mcT_M$ the analytic tangent sheaf of $M$.

\begin{definition}
A flat F-manifold $(M,\nabla,\circ)$ is called a Frobenius manifold if the tangent spaces~$T_pM$ are equipped with a symmetric non-degenerate bilinear form~$\eta$ analytically depending on the point $p\in M$ such that $\nabla \eta = 0$ and for any $X,Y,Z \in \T_M$ the following condition is satisfied:
$$
\eta(X\circ Y, Z) = \eta(X,Y\circ Z).
$$
The connection $\nabla$ is then the Levi-Civita connection associated to the form~$\eta$. A Frobenius manifold will be denoted by the triple $(M,\eta,\circ)$. The form $\eta$ is traditionally called a metric.
\end{definition}

Let $(M,\eta,\circ)$ be a Frobenius manifold and consider the flat coordinates $t^1,\ldots,t^N$ of the metric~$\eta$ and the vector potential $(F^1,\ldots,F^N)$. Then locally there exists an analytic function~$F$ such that $F^\alpha = \eta^{\alpha\beta}\frac{\d F}{\d t^\beta}$ and $\frac{\d^3 F}{\d t^1\d t^\alpha\d t^\beta}=\eta_{\alpha\beta}$, where $(\eta_{\alpha\beta})$ is the matrix of the form $\eta$ in the coordinates $t^1,\ldots,t^N$. The function $F$ satisfies the WDVV equations~\eqref{eq:WDVV equations} and is called the Frobenius manifold potential.

A Frobenius manifold $(M,\eta,\circ)$ is called conformal if the corresponding flat F-manifold is conformal and the metric $\eta$ satisfies the condition
$$
\mcL_E\eta=(2-\delta)\eta,\quad\text{for some $\delta\in\mbC$},
$$
where $\mcL_E$ denotes the Lie derivative. The number $\delta$ is then called the conformal dimension of the Frobenius manifold. The Frobenius manifold potential $F$ satisfies then the condition
$$
E(F) = (3-\delta) F + \frac{1}{2}A_{\alpha\beta}t^\alpha t^\beta+B_\alpha t^\alpha+C,\quad\text{for some $A_{\alpha\beta},B_\alpha,C\in\mbC$}.
$$
In the theory of Frobenius manifolds it is typically assumed that one can choose flat coordinates such that the matrix $\left(\frac{\d E^\alpha}{\d t^\beta}\right)$ is diagonal and so the Euler vector field has form~\eqref{eq:Euler field for Frobenius}.

The papers~\cite{Dub96,Dub99} contain a systematic study of the theory of Frobenius manifolds.

\subsection{Extensions of flat F-manifolds and the open WDVV equations}\label{subsection:flat F-manifolds and open WDVV}

Consider a flat F-manifold structure, given by a vector potential $(F^1,\ldots,F^{N+1})$ on an open subset ${M\times U\in\mbC^{N+1}}$, where $M$ and $U$ are open subsets of $\mbC^N$ and $\mbC$, respectively. Suppose that the functions $F^1,\ldots,F^N$ don't depend on the variable $t^{N+1}$, varying in $U$. Then the functions $F^1,\ldots,F^N$ satisfy equations~\eqref{eq:axiom2 of flat F-man} and, thus, define a flat F-manifold structure on~$M$. In this case we call the flat F-manifold structure on $M\times U$ an {\it extension} of the flat F-manifold structure on $M$.

Consider the flat F-manifold associated to a Frobenius manifold, given by a potential $F(t^1,\ldots,t^N)\in\mcO(M)$ and a metric $\eta$, $F^\alpha=\eta^{\alpha\mu}\frac{\d F}{\d t^\mu}$, $1\le\alpha\le N$. It is easy to check that a function $F^o(t^1,\ldots,t^N,s)\in\mcO(M\times U)$ satisfies equations~\eqref{eq:open WDVV,1},~\eqref{eq:open WDVV,2} and~\eqref{eq:unit condition for Fo} if and only if the $(N+1)$-tuple $(F^1,\ldots,F^N,F^o)$ is a vector potential of a flat F-manifold. Here we identify $s=t^{N+1}$. This defines a correspondence between solutions of the open WDVV equations, satisfying property~\eqref{eq:unit condition for Fo}, and flat F-manifolds, extending the Frobenius manifold given. This observation belongs to Paolo Rossi. 


\section{Saito Frobenius manifolds}

In this section we recall the Dubrovin--Saito construction of a Frobenius manifold structure on the parameter space of a universal unfolding of a simple singularity. 

Let us first recall the list of polynomials defining simple singularities:
\begin{align*}
f_{A_N}(x,y)=&\frac{x^{N+1}}{N+1} + y^2,&& N\ge 1,\\
f_{D_N}(x,y)=&\frac{x^{N-1}}{N-1} + x y^2,&& N\ge 4,\\
f_{E_6}(x,y)=&x^4 + y^3,\\
f_{E_7}(x,y)=&x^3y + y^3,\\
f_{E_8}(x,y)=&x^5 + y^3.
\end{align*}
The associated {\it local algebra} is defined by 
$$
\mcA_W:=\mbC[x,y]\left/\left(\frac{\d f_W}{\d x},\frac{\d f_W}{\d y}\right)\right.,
$$
where $W$ is one of the types $A_N$, $D_N$ or $E_N$. Because $x=y=0$ is an isolated critical point of~$f_W$, the local algebra $\mcA_W$ turns out to be a finite-dimensional vector space. Denote $\dim\mcA_W=:N$. A {\it universal unfolding} of the singularity of $f_W$ is a function $\Lambda_W\colon \CC^{2}\times\CC^{N} \to \CC$ of the form
\begin{gather*}
\Lambda_W(x,y,v_1,\dots,v_N)=f_W+\sum_{i=1}^N v_i\phi_i,\quad\phi_i\in\mbC[x,y],
\end{gather*}    
such that the classes of polynomials $\phi_1,\ldots,\phi_N$ form a basis of the local algebra $\mcA_W$. Explicitly, universal unfoldings of the ADE singularities are given by 
\begin{align*}
  \Lambda_{A_N} &= \frac{x^{N+1}}{N+1} + y^2+ \sum_{k=1}^N v_k x^{k-1},\\
  \Lambda_{D_N} &= \frac{x^{N-1}}{N-1} + x y^2 + \sum _{k=1}^{N-1} v_k x^{k-1} + v_{N}y,\\
  \Lambda_{E_6} &= x^4 + y^3 + v_1 + v_2 x + v_3 y + v_4 x^2 + v_5 x y + v_6 x^2 y,\\
  \Lambda_{E_7} &= x^3y + y^3 + v_1 + v_2 x + v_3 y + v_4 x^2 + v_5 x y + v_6 x^3  + v_7 x^4,\\
  \Lambda_{E_8} &= x^5 + y^3 + v_1 + v_2 x + v_3 y + v_4 x^2 + v_5 x y + v_6 x^3  + v_7 x^2y + v_8 x^3y.
\end{align*}

Consider the quotient ring
$$
\hmcA_W:=\CC[x,y,v_1,\dots,v_N]/\left(\p_x\Lambda_W,\p_y\Lambda_W\right).
$$
As a $\mbC[v_1,\ldots,v_N]$-module, the space $\hmcA_W$ has dimension $N$ with a basis given by the classes $[\phi_1],\ldots,[\phi_N]\in\hmcA_W$ of the polynomials $\phi_1,\ldots,\phi_N$. Identifying the $\mbC[v_1,\ldots,v_N]$-modules $\mcT_{\mbC^N}^\alg(\mbC^N)$ and $\hmcA_W$ via the isomorphism $\Psi_W$ defined by
\begin{gather*}
\Psi_W\left(\frac{\p}{\p v_k}\right) := \left[\phi_k\right], \quad 1 \le k \le N,
\end{gather*}
by Remark~\ref{remark:algebraic construction of multiplication}, we endow the tangent spaces $T_p \mbC^N$ with a multiplication. The structure constants of it are polynomials in the coordinates $v_1,\ldots,v_N$.

A flat metric can be introduced as follows. It is easy to see that there exist unique positive rational numbers $q_x,q_y,q_1,\ldots,q_N$ such that
$$
q_x x \frac{\p\Lambda_W}{\p x} + q_y y \frac{\p\Lambda_W}{\p y} + \sum_{k=1}^N q_k v_k \frac{\p\Lambda_W}{\d v_k}=\Lambda_W.
$$
There is a unique index $1\le l\le N$ such that the number~$q_l$ is minimal. For the singularities~$A_N$ and~$E_N$ we have $l=N$ and in the $D_N$-case we have $l=N-1$. Denote by $(c_v)_{i,j}^k$ the structure constants of multiplication on $\mbC^N$ in the coordinates $v_1,\ldots,v_N$. Define a bilinear form $\eta_W$ on~$\mbC^N$ by
$$
\eta_W\left(\frac{\d}{\d v_i},\frac{\d}{\d v_j}\right):=(c_v)^l_{i,j}.
$$ 
This bilinear form together with the multiplication, introduced above, define a Frobenius manifold structure on $\mbC^N$, often called the Saito Frobenius manifold. It is semisimple and conformal with an Euler vector field $E_W$ given by
$$
E_W=\sum_{k=1}^N q_k v_k \frac{\p}{\d v_k}.
$$
The conformal dimension is $\delta=1-q_l$. 
\begin{remark}
The coordinates $v_1,\dots,v_N$ are not flat whenever $N \ge 3$. 
\end{remark}

There exist unique {\it global} flat coordinates $t^i(v_1,\ldots,v_N)$ on $\mbC^N$ such that
$$
t^i(v_1,\ldots,v_N)=v_i+O(v_*^2).
$$
They satisfy the quasi-homogeneity condition
\begin{gather}\label{eq:homogeneity of flat coordinates}
E_W(t^i(v_1,\ldots,v_N))=q_i t^i(v_1,\ldots,v_N),
\end{gather}
and, hence, the Euler vector field in the flat coordinates $t^i$ is given by
$$
E_W=\sum_{i=1}^N q_i t^i\frac{\d}{\d t^i}.
$$
The Frobenius manifold structure in the flat coordinates $t^i$ is described by a polynomial potential $F_W(t^1,\ldots,t^N)\in\mbC[t^1,\ldots,t^N]$, which we fix by requiring that it doesn't contain monomials of degree less than $3$. Then the polynomial $F_W$ satisfies the condition
\begin{gather}\label{eq:homogeneity of potential for Coxeter}
E_W(F_W)=(3-\delta)F_W.
\end{gather}

Explicit formulas for the flat coordinates of the Saito Frobenius manifolds of simple singularities are given in~\cite{NY98}. For the $A_N$-case the formula is
$$
t^\gamma(v_1,\ldots,v_N) = \hspace{-0.2cm}\sum_{\substack{\alpha_1,\ldots,\alpha_N\ge 0\\ \sum (N+2-i)\alpha_i = N+2-\gamma}}\hspace{-0.2cm}\frac{1}{N+1-\gamma}\prod_{k=0}^{\sum\alpha_i-1}\left(N+1-\gamma-k(N+1)\right)\frac{\prod v_i^{\alpha_i}}{\prod \alpha_i!}, \quad 1\le\gamma \le N.
$$
For the $D_N$-case the formula is
\begin{align}
&t^\gamma(v_1,\ldots,v_N) = \hspace{-0.2cm}\sum_{\substack{\alpha_1,\ldots,\alpha_{N-1}\ge 0\\ \sum (N-i)\alpha_i = N-\gamma}}\hspace{-0.1cm} \left(-\frac{1}{2}\right)^{\sum\alpha_i-1} \prod_{k=0}^{\sum\alpha_i-2} \left(2\gamma-1 + 2k(N-1) \right) \frac{\prod v_i^{\alpha_i}}{\prod \alpha_i!}, \quad 1\le\gamma \le N-1,\label{eq:Dn flat coordinates}\\
&t^N(v_1,\ldots,v_N)=v_N.\notag
\end{align}
Note that in this case the coordinates $t^1,\ldots,t^{N-1}$ depend only on $v_1,\ldots,v_{N-1}$.


\section{Extended $r$-spin theory and the open WDVV equations for the $A$-singularity}\label{section:extended r-spin}

Here we explain how the Saito Frobenius manifold of the $A$-singularity together with a certain solution of the open WDVV equations appear in the intersection theory on the moduli spaces of algebraic curves.  

Let $r=N+1$. For integers $0\le\alpha_1,\ldots,\alpha_n\le r-1$, using the geometry of algebraic curves with an $r$-spin structure, one can construct a cohomology class
$$
W^r_{0,n}(\alpha_1,\ldots,\alpha_n)\in H^{2d}(\oM_{0,n},\mbQ),\quad d=\frac{\sum\alpha_i-(r-2)}{r},
$$ 
called {\it Witten's class}, on the moduli space $\oM_{0,n}$ of stable curves of genus $0$ with $n$ marked points (see e.g.~\cite{PPZ16}). Here we assume that the class $W^r_{0,n}(\alpha_1,\ldots,\alpha_n)$ is equal to zero, if~$d$ is not an integer or $d<0$. The class $W^r_{0,n}(\alpha_1,\ldots,\alpha_n)$ vanishes, if one of the $\alpha_i$'s is equal to~$r-1$. Consider the generating series 
$$
F_{\rspin}(t^1,\ldots,t^{r-1}):=\sum_{n\ge 3}\frac{1}{n!}\sum_{0\le\alpha_1,\ldots,\alpha_n\le r-2}\left(\int_{\oM_{0,n}}W^r_{0,n}(\alpha_1,\ldots,\alpha_n)\right)\prod_{i=1}^n t^{\alpha_i+1}.
$$
The functions $F_{A_N}$ and $F_{\rspin}$ are related by~\cite{JKV01a}
$$
F_{A_N}(t^1,\ldots,t^N)=(-r)^{-3}F_{\rspin}((-r)t^1,\ldots,(-r)t^N).
$$
This is one of the simplest cases of mirror symmetry. Denote
$$
\<\tau_{\alpha_1}\ldots\tau_{\alpha_n}\>_{A_N}:=\left.\frac{\d^n F_{A_N}}{\d t^{\alpha_1}\ldots\d t^{\alpha_n}}\right|_{t^*=0},\quad 1\le\alpha_1,\ldots,\alpha_n\le N.
$$
We have (see e.g.~\cite[page 4]{LVX17})
\begin{gather}\label{eq:3-point and 4-point for AN}
\<\tau_{\alpha_1}\tau_{\alpha_2}\tau_{\alpha_3}\>_{A_N}=\delta_{\alpha_1+\alpha_2+\alpha_3,N+2},\qquad\<\tau_{\alpha_1}\tau_{\alpha_2}\tau_{\alpha_3}\tau_{\alpha_4}\>_{A_N}=-\min(\alpha_i-1,N+1-\alpha_i).
\end{gather}
These formulas will be used later.

In~\cite{JKV01b} the authors noticed that the construction of Witten's class $W^r_{0,n}(\alpha_1,\ldots,\alpha_n)$ can be extended to the case when $\alpha_1=-1$ and $0\le\alpha_2,\ldots,\alpha_n\le r-1$. In~\cite{BCT19} the authors considered the generating series 
$$
F^\ext_{\rspin}(t^1,\ldots,t^r):=\sum_{n\ge 2}\frac{1}{n!}\sum_{0\le\alpha_1,\ldots,\alpha_n\le r-1}\left(\int_{\oM_{0,n+1}}W^r_{0,n+1}(-1,\alpha_1,\ldots,\alpha_n)\right)\prod_{i=1}^n t^{\alpha_i+1}
$$
and proved that it satisfies the open WDVV equations, associated to the potential $F_{\rspin}$, together with property~\eqref{eq:unit condition for Fo}. Here one should identify $t^r=s$. It occurs that after a simple transformation the function $F^\ext_{\rspin}$ also gives the generating series of intersection numbers on the moduli space of $r$-spin disks, introduced in~\cite[Theorem~1.3]{BCT18}.

Introduce a function $F^o_{A_N}(t^1,\ldots,t^N,s)$ by
$$
F^o_{A_N}(t^1,\ldots,t^N,s):=(-r)^{-2}F^\ext_{\rspin}((-r)t^1,\ldots,(-r)t^N,(-r)s).
$$
Clearly, it satisfies the open WDVV equations, associated to $F_{A_N}$, together with condition~\eqref{eq:unit condition for Fo}. In~\cite[Proposition~5.1]{BCT18} the authors found an explicit formula for the coefficients of $F^\ext_{\rspin}$ that gives
\begin{gather*}
\<\tau_{\alpha_1}\ldots\tau_{\alpha_n}\sigma^k\>^o_{A_N}=
\begin{cases}
(n+k-2)!,&\text{if $\sum_{i=1}^n(N+2-\alpha_i)+k=N+2$},\\
0,&\text{otherwise},
\end{cases}
\end{gather*}
where we use the notation
$$
\<\tau_{\alpha_1}\ldots\tau_{\alpha_n}\sigma^k\>^o_{A_N}:=\left.\frac{\d^{n+k} F^o_{A_N}}{\d t^{\alpha_1}\ldots\d t^{\alpha_n}\d s^k}\right|_{t^*=s=0}.
$$
In~\cite{Bur18} the author proved that the coefficients of the function $F^o_{A_N}$ are related to the expression of the coordinates $v_k$ in the terms of the flat coordinates $t^1,\ldots,t^N$ on the Saito Frobenius manifold of the singularity $A_N$ by
\begin{gather}\label{eq:FoAN and the flat coordinates}
\frac{\d F^o_{A_N}}{\d s}=\frac{s^{N+1}}{N+1}+\sum_{k=1}^Ns^{k-1}v_k(t^1,\ldots,t^N).
\end{gather}


\section{Generalized Dubrovin--Saito construction for the singularities $A$ and $D$}\label{section: Mwext explicit}

In this section, for the singularites of types $A$ and $D$ we present a generalization of the Dubrovin--Saito construction that produces a flat F-manifold that extends the Saito Frobenius manifold and, therefore, gives a solution of the open WDVV equations. In the $A_N$-case this solution coincides with the function $F^o_{A_N}$. In both $A$- and $D$-cases, the coefficients of powers of the variable~$s$ in this solution coincide with the transition functions between two coordinate systems on the Saito Frobenius manifold. 

\subsection{$A_N$-case}

Consider the space $M^\ext_{A_N}:=\mbC^{N+1}$ with coordinates $v_1,\ldots,v_{N+1}$. Consider the quotient ring
$$
\hmcA^{\ext}_{A_N}:=\CC[x,y,w,v_1,\ldots,v_{N+1}] \Big/ \left(w - \p_x\Lambda_{A_N}, \p_y \Lambda_{A_N}, wx - v_{N+1} w \right).
$$
As a $\mbC[v_1,\ldots,v_{N+1}]$-module, the space $\hmcA^{\ext}_{A_N}$ is free of dimension $N+1$, and the elements $[1],[x],\ldots,[x^{N-1}],[w]$ form a basis. To show that any other element $[x^i y^j w^k]$ can be expressed in terms of them, first note that, obviously, $[y]=0$ and $[w x]=v_{N+1}[w]$. We also see that
\begin{gather}\label{eq:w2 for AN}
[w^2]=\left[w\left(x^N+\sum_{k=2}^{N}(k-1)v_k x^{k-2}\right)\right]=\left(v_{N+1}^N+\sum_{k=2}^N(k-1)v_k v_{N+1}^{k-2}\right)[w].
\end{gather}
Using the relation $[x^N]=[w]-\sum_{k=2}^N(k-1)v_k[x^{k-2}]$, we can express any element $[x^p]$ with $p\ge N$ in terms of the elements $[1],[x],\ldots,[x^{N-1}],[w]$.

Identifying the $\mbC[v_1,\ldots,v_{N+1}]$-modules $\mcT^\alg_{M^\ext_{A_N}}(M^\ext_{A_N})$ and $\hmcA^\ext_{A_N}$ via the isomorphism $\Psi^\ext_{A_N}$ defined by
\begin{gather*}
\Psi_W^\ext\left(\frac{\p}{\p v_k}\right) := [x^{k-1}], \quad 1 \le k \le N,\qquad \Psi_W^\ext \left(\frac{\p}{\p v_{N+1}} \right) := [w],
\end{gather*}
by Remark~\ref{remark:algebraic construction of multiplication}, we endow the tangent spaces $T_pM_{A_N}^\ext$ with a multiplication and, clearly, the structure constants of it are polynomials in the coordinates $v_1,\ldots,v_{N+1}$.

Consider the flat coordinates $t^\alpha=t^\alpha(v_1,\dots,v_N)$, $1 \le \alpha \le N$, the potential $F_{A_N}(t^1,\ldots,t^N)$ of the Frobenius manifold of the singularity $A_N$ and the function~$F^o_{A_N}$, described in Section~\ref{section:extended r-spin}.

\begin{theorem}~\\
1. The coordinates $t^1(v_1,\ldots,v_N),\ldots,t^N(v_1,\ldots,v_N)$ and $t^{N+1}:=v_{N+1}$ together with the multiplicative structure on $M^\ext_{A_N}$, constructed above, define a flat F-manifold structure on $M^\ext_{A_N}$.\\
2. The vector potential of this flat F-manifold is given by $\left(\eta_{A_N}^{1\alpha} \frac{\d F_{A_N}}{\d t^\alpha}, \ldots, \eta_{A_N}^{N\alpha} \frac{\d F_{A_N}}{\d t^\alpha}, F^o_{A_N}\right)$, where we identify $s=t^{N+1}$.
\end{theorem}
\begin{proof}
We denote by $(c^\ext)^\alpha_{\beta\gamma}$ the structure constants of multiplication in the coordinates $t^1,\ldots,t^{N+1}$ and by $(c^\ext_v)^\alpha_{\beta\gamma}$ the structure constants of multiplication in the coordinates $v_1,\ldots,v_{N+1}$.

In order to prove the theorem, we have to check the following equations:
\begin{align}
(c^\ext)^\alpha_{\beta\gamma}=&\sum_{\mu=1}^N\eta_{A_N}^{\alpha\mu} \frac{\d^3 F_{A_N}}{\d t^\mu\d t^\beta\d t^\gamma},&& 1\le\alpha\le N,&& 1\le\beta,\gamma\le N+1,\label{eq:first property for AN}\\
(c^\ext)^{N+1}_{\alpha\beta}=&\frac{\d^2 F^o_{A_N}}{\d t^\alpha\d t^\beta},&& 1\le\alpha,\beta\le N+1.&&\label{eq:second property for AN}
\end{align}
Since the subspace $\mbC[v_1,\ldots,v_{N+1}]\<[w]\>$ is an ideal in the ring $\hmcA^{\ext}_{A_N}$ and the quotient by this ideal coincides with the ring $\hmcA_{A_N}$, we have
$$
(c^\ext_v)^a_{b,c}=
\begin{cases}
(c_v)^a_{b,c},&\text{if $1\le a,b,c\le N$},\\
0,&\text{if $1\le a\le N$ and one of the indices $b,c$ is equal to $N+1$}.
\end{cases}
$$
This implies equation~\eqref{eq:first property for AN} and it remains to prove~\eqref{eq:second property for AN}. 

Suppose $1\le\alpha\le N$ and $\beta=N+1$. Since $[xw]=v_{N+1}[w]$, we have $(c^\ext_v)^{N+1}_{k,N+1}=v_{N+1}^{k-1}$ for $1\le k\le N$, and, therefore,
$$
(c^\ext)^{N+1}_{\alpha,N+1}=\sum_{k=1}^N\frac{\d v_k}{\d t^\alpha}(c^\ext_v)^{N+1}_{k,N+1}=\sum_{k=1}^N\frac{\d v_k}{\d t^\alpha}v_{N+1}^{k-1}\stackrel{\text{eq.\eqref{eq:FoAN and the flat coordinates}}}{=}\frac{\d^2 F^o_{A_N}}{\d t^\alpha\d t^{N+1}}.
$$

Suppose $\alpha=\beta=N+1$. Then we compute
$$
(c^\ext)^{N+1}_{N+1,N+1}=(c^\ext_v)^{N+1}_{N+1,N+1}\stackrel{\text{eq.\eqref{eq:w2 for AN}}}{=}v_{N+1}^N+\sum_{k=2}^N(k-1)v_k v_{N+1}^{k-2}\stackrel{\text{eq.\eqref{eq:FoAN and the flat coordinates}}}{=}\frac{\d^2 F^o_{A_N}}{\d (t^{N+1})^2}.
$$

Finally, if $1\le\alpha,\beta\le N$, then from the associativity of the algebra $\hmcA^{\ext}_{A_N}$ we get
$$
(c^{\ext})^{N+1}_{\alpha\beta}(c^\ext)^{N+1}_{N+1,N+1}=(c^\ext)^{N+1}_{\alpha,N+1}(c^\ext)^{N+1}_{\beta,N+1}-\sum_{\mu=1}^N c^\mu_{\alpha\beta}(c^\ext)^{N+1}_{\mu,N+1}.
$$
Since the function $F^o_{A_N}$ satisfies~\eqref{eq:open WDVV,2} and, as we have just proved, $(c^\ext)^{N+1}_{\gamma,N+1}=\frac{\d^2 F^o_{A_N}}{\d t^\gamma\d t^{N+1}}$ for $1\le\gamma\le N+1$, we obtain $(c^\ext)^{N+1}_{\alpha\beta}=\frac{\d^2 F^o_{A_N}}{\d t^\alpha\d t^\beta}$.
\end{proof}

\subsection{$D_N$-case}

Consider the space $M^\ext_{D_N}:=\mbC^N\times\mbC^*$ with coordinates $v_1,\ldots,v_{N+1}$. Consider the quotient ring
\begin{gather*}
\hmcA^{\ext}_{D_N} := \CC[x,y,w,v_1,\dots,v_N,v_{N+1},v_{N+1}^{-1}]\Big/ \left(w - v_{N+1} \p_x\Lambda_{D_N},\p_y \Lambda_{D_N},2wx - v_{N+1}^2w\right).
\end{gather*}
As a $\CC[v_1,\dots,v_N,v_{N+1},v_{N+1}^{-1}]$-module, the space $\hmcA^{\ext}_{D_N}$ is free of dimension $N+1$ with a basis $[1],[x],\ldots,[x^{N-2}],[y],[w]$. To show that any other element $[x^i y^j w^k]$ can be expressed in terms of them, first note that 
\begin{gather}\label{eq:xy for DN}
[xy]=-\frac{v_N}{2}[1],\qquad[w x]=\frac{v_{N+1}^{2}}{2}[w],\qquad [wy]=-\frac{v_N}{v_{N+1}^2}[w],
\end{gather}
where the last equation follows from the first two. Similarly to the $A_N$-case, we have
\begin{gather}\label{eq:w2 for DN}
[w^2]=\left(\frac{v_{N+1}^{2N-3}}{2^{N-2}}+\frac{v_N^2}{v_{N+1}^3}+\sum_{k=2}^{N-1}(k-1)v_k \frac{v_{N+1}^{2k-3}}{2^{k-2}}\right)[w].
\end{gather}
Using that $[w-v_{N+1} \p_x\Lambda_{D_N}]=0$, we obtain
\begin{gather}\label{eq:y2 for DN}
[y^2]=\frac{1}{v_{N+1}}[w]-[x^{N-2}]-\sum_{k=2}^{N-1}(k-1)v_k[x^{k-2}].
\end{gather}
Multiplying this equation by $[x]$, we get the relation
\begin{gather}\label{eq:xN-1 for DN}
[x^{N-1}]=\frac{v_{N+1}}{2}[w]+\frac{v_N}{2}[y]-\sum_{k=1}^{N-1}(k-1)v_k[x^{k-1}],
\end{gather}
that allows to express any element $[x^p]$ with $p\ge N-1$ in terms of the elements $[1],[x],\ldots,[x^{N-2}]$, $[y]$, $[w]$.

Identifying the $\CC[v_1,\dots,v_N,v_{N+1},v_{N+1}^{-1}]$-modules $\mcT^\alg_{M^\ext_{D_N}}(M^\ext_{D_N})$ and $\hmcA^\ext_{D_N}$ via the isomorphism~$\Psi^\ext_{D_N}$ defined by
\begin{gather*}
\Psi_W^\ext\left(\frac{\p}{\p v_k}\right) := [x^{k-1}], \quad 1 \le k \le N-1,\qquad \Psi_W^\ext \left(\frac{\p}{\p v_N} \right) := [y],\qquad \Psi_W^\ext \left(\frac{\p}{\p v_{N+1}} \right) := [w],
\end{gather*}
we endow the tangent spaces $T_pM_{D_N}^\ext$ with a multiplication and, clearly, the structure constants of it belong to the ring $\mbC[v_1,\ldots,v_N,v_{N+1},v_{N+1}^{-1}]$.

Consider the flat coordinates $t^\alpha=t^\alpha(v_1,\dots,v_N)$, $1 \le \alpha \le N$, and the potential $F_{D_N}(t^1,\ldots,t^N)$ of the Frobenius manifold of the singularity $D_N$. Let $t^{N+1}:=v_{N+1}$ and define a function $F^o_{D_N}(t^1,\ldots,t^{N+1})$ by
\begin{gather*}
F^o_{D_N}:=\left.\left(\sum_{k=1}^{N-1} \frac{v_k v_{N+1}^{2 k-1}}{2^{k-1}(2 k-1)} + \frac{v_{N+1}^{2 N-1}}{2^{N-2}(2 N-1) (2N-2)} + \frac{v_N^2}{2 v_{N+1}}\right)\right|_{v_i=v_i(t^*)}.
\end{gather*}

\begin{theorem}~\\
1. The coordinates $t^1(v_1,\ldots,v_N),\ldots,t^N(v_1,\ldots,v_N)$ and $t^{N+1}=v_{N+1}$ together with the multiplicative structure on $M^\ext_{D_N}$, constructed above, define a flat F-manifold structure on $M^\ext_{D_N}$.\\
2. The vector potential of this flat F-manifold is given by $\left(\eta_{D_N}^{1\alpha} \frac{\d F_{D_N}}{\d t^\alpha}, \ldots, \eta_{D_N}^{N\alpha} \frac{\d F_{D_N}}{\d t^\alpha}, F^o_{D_N}\right)$.
\end{theorem}
\begin{proof}
We denote by $(c^\ext)^\alpha_{\beta\gamma}$ the structure constants of multiplication in the coordinates $t^1,\ldots,t^{N+1}$ and by $(c^\ext_v)^\alpha_{\beta\gamma}$ the structure constants of multiplication in the coordinates $v_1,\ldots,v_{N+1}$.

In order to prove the theorem, we have to check the following equations:
\begin{align}
(c^\ext)^\alpha_{\beta\gamma}=&\sum_{\mu=1}^N\eta_{D_N}^{\alpha\mu} \frac{\d^3 F_{D_N}}{\d t^\mu\d t^\beta\d t^\gamma},&& 1\le\alpha\le N,&& 1\le\beta,\gamma\le N+1,\label{eq:first property for DN}\\
(c^\ext)^{N+1}_{\alpha\beta}=&\frac{\d^2 F^o_{D_N}}{\d t^\alpha\d t^\beta},&&1\le\alpha,\beta\le N+1.&&\label{eq:second property for DN}
\end{align}
Since the subspace $\mbC[v_1,\ldots,v_N,v_{N+1},v_{N+1}^{-1}]\<[w]\>$ is an ideal in the ring $\hmcA^{\ext}_{D_N}$ and the quotient by this ideal coincides with the ring $\hmcA_{D_N}$, we have
$$
(c^\ext_v)^a_{b,c}=
\begin{cases}
(c_v)^a_{b,c},&\text{if $1\le a,b,c\le N$},\\
0,&\text{if $1\le a\le N$ and one of the indices $b,c$ is equal to $N+1$}.
\end{cases}
$$
This implies equation~\eqref{eq:first property for DN} and it remains to prove~\eqref{eq:second property for DN}. 

We have two substantially different cases: the case $\alpha\in\{N,N+1\}$ or $\beta\in\{N,N+1\}$ and the case $\alpha,\beta\in\{1,2,\dots,N-1\}$.

\bigskip

\noindent{\it Case $\alpha\in\{N,N+1\}$ or $\beta\in\{N,N+1\}$}. If $\alpha,\beta\in\{N,N+1\}$, then 
$(c^\ext)^{N+1}_{\alpha,\beta}=(c^\ext_v)^{N+1}_{\alpha,\beta}$ and $\frac{\d^2 F^o_{D_N}}{\d t^\alpha\d t^\beta}=\frac{\d^2 F^o_{D_N}}{\d v_\alpha\d v_\beta}$. The equation $(c^\ext_v)^{N+1}_{\alpha,\beta}=\frac{\d^2 F^o_{D_N}}{\d v_\alpha\d v_\beta}$ immediately follows from formulas~\eqref{eq:xy for DN},~\eqref{eq:w2 for DN} and~\eqref{eq:y2 for DN}.

If $1\le\alpha\le N-1$ and $\beta=N+1$, then
\begin{gather*}
(c^\ext)^{N+1}_{\alpha,N+1}=\sum_{k=1}^{N-1}\frac{\d v_k}{\d t^\alpha}(c^\ext_v)^{N+1}_{k,N+1}\stackrel{\text{eq.\eqref{eq:xy for DN}}}{=}\sum_{k=1}^{N-1}\frac{\d v_k}{\d t^\alpha}\frac{v_{N+1}^{2k-2}}{2^{k-1}}=\frac{\d}{\d t^\alpha}\frac{\d F^o_{D_N}}{\d v_{N+1}}=\frac{\d^2 F^o_{D_N}}{\d t^\alpha\d t^{N+1}}.
\end{gather*}

If $1\le\alpha\le N-1$ and $\beta=N$, then
\begin{gather*}
(c^\ext)^{N+1}_{\alpha,N}=\sum_{k=1}^{N-1}\frac{\d v_k}{\d t^\alpha}(c^\ext_v)^{N+1}_{k,N}\stackrel{\text{eq.\eqref{eq:xy for DN}}}{=}0=\frac{\d}{\d t^\alpha}\frac{\d F^o_{D_N}}{\d v_N}=\frac{\d^2 F^o_{D_N}}{\d t^\alpha\d t^N}.
\end{gather*}

\bigskip

\noindent{\it Case $\alpha,\beta\in\{1,2,\ldots,N-1\}$}. We have to check that
\begin{gather*}
(c^\ext)^{N+1}_{\alpha\beta}=\frac{\d^2 F^o_{D_N}}{\d t^\alpha\d t^\beta} \Leftrightarrow \sum_{1\le a,b\le N-1}\frac{\d v_a}{\d t^\alpha}\frac{\d v_b}{\d t^\beta}(c^\ext_v)^{N+1}_{a,b}=\sum_{k=1}^{N-1}\frac{\d^2 v_k}{\d t^\alpha\d t^\beta}\frac{v_{N+1}^{2k-1}}{2^{k-1}(2k-1)},
\end{gather*}
that is equivalent to the equation
\begin{gather}\label{eq:equation for cab for DN}
(c^\ext_v)^{N+1}_{a,b}=-\sum_{k=1}^{N-1}\frac{\d v_k}{\d t^\gamma}\frac{\d^2 t^\gamma}{\d v_a\d v_b}\frac{v_{N+1}^{2k-1}}{2^{k-1}(2k-1)},\quad 1\le a,b\le N-1.
\end{gather}

Let us compute the structure constants $(c^\ext_v)^{N+1}_{a,b}$. Introduce polynomials $\omega_k\in\mbQ[v_1,\ldots,v_{N-1}]$, $k \ge 0$, by
\begin{gather*}
\omega_k := \sum_{\substack{\alpha_1,\dots,\alpha_{N-1}\ge 0\\ \sum (N-i)\alpha_i = k}} s_1^{\alpha_1}\cdots s_{N-1}^{\alpha_{N-1}} \frac{(\sum \alpha_i )!}{\prod \alpha_i!},\quad\text{where $s_i:=(1-i)v_i$ for $1\le i\le N-1$.}
\end{gather*}
The first few functions $\omega_k$ are
\[
  \omega_0 = 1, \quad \omega_1 = s_{N-1}, \quad \omega_2 = s_{N-2} + s_{N-1}^2, \quad \omega_3 = s_{N-3} + 2 s_{N-2}s_{N-1} + s_{N-1}^3.
\]
The functions $\omega_k$ satisfy the recursion relation
\begin{gather}\label{eq:recursion for omegak}
\omega_{k+1}=\sum_{i=1}^{N-1}s_{N-i}\omega_{k+1-i},\quad k\ge 0,
\end{gather}
where we adopt the convention $\omega_j:=0$ for $j<0$.

\begin{lemma}\label{lemma: Dn xp product}
For $1\le a,b\le N-1$ we have
\begin{gather}\label{eq:cabN+1 in Dn}
c_{a,b}^{N+1} = \sum_{k=0}^{a+b-N-1} \omega_k \frac{v_{N+1}^{2(a+b)-2N-1-2k}}{2^{a+b-N-k}}.
\end{gather}
\end{lemma}
\begin{proof}
Equation~\eqref{eq:cabN+1 in Dn} is equivalent to the following formula:
\begin{equation}\label{eq:Dn cab classes}
\Coef_{[w]}[x^{p-2}] = \sum_{k=0}^{p-N-1} \omega_k \frac{v_{N+1}^{2p-2N-1-2k}}{2^{p-N-k}},\quad 2\le p\le 2N-2,
\end{equation}
where $\Coef_{[w]}[x^{p-2}]$ denotes the coefficient of $[w]$ in the expression for $[x^{p-2}]$ in terms of the basis elements $[1],[x],\ldots,[x^{N-2}],[y],[w]$. For $p\le N$ formula~\eqref{eq:Dn cab classes} is obvious because both sides of it are equal to zero. 

Suppose $p\ge N+1$. Multiplying both sides of equation~\eqref{eq:xN-1 for DN} by $[x^{p-N-1}]$, we get the relation
$$
\Coef_{[w]}[x^{p-2}]=\frac{v_{N+1}^{2p-2N-1}}{2^{p-N}}+\sum_{k=1}^{N-1}s_k\Coef_{[w]}[x^{p+k-N-2}],
$$
that allows to compute the coefficients $\Coef_{[w]}[x^{p-2}]$ recursively. Then, using also relation~\eqref{eq:recursion for omegak}, formula~\eqref{eq:Dn cab classes} can be easily proved by induction.
\end{proof}

Using the lemma, we see that equation~\eqref{eq:equation for cab for DN} can be equivalently written as
\begin{align*}
&\sum_{k=0}^{a+b-N-1} \omega_k \frac{v_{N+1}^{2(a+b)-2N-1-2k}}{2^{a+b-N-k}}=-\sum_{k=1}^{N-1}\frac{\d v_k}{\d t^\gamma}\frac{\d^2 t^\gamma}{\d v_a\d v_b}\frac{v_{N+1}^{2k-1}}{2^{k-1}(2k-1)}\Leftrightarrow\\
\Leftrightarrow&\sum_{k=1}^{N-1} \omega_{a+b-N-k} \frac{v_{N+1}^{2k-1}}{2^k}=-\sum_{k=1}^{N-1}\frac{\d v_k}{\d t^\gamma}\frac{\d^2 t^\gamma}{\d v_a\d v_b}\frac{v_{N+1}^{2k-1}}{2^{k-1}(2k-1)}.
\end{align*}
So we have to prove that
\begin{gather}\label{eq:identity for DN}
\frac{\d v_k}{\d t^\gamma}\frac{\d^2 t^\gamma}{\d v_a\d v_b}=-\frac{2k-1}{2}\omega_{a+b-N-k}\Leftrightarrow \frac{\d^2 t^\gamma}{\d v_a\d v_b}=-\sum_{k=1}^{N-1}\frac{2k-1}{2}\omega_{a+b-N-k}\frac{\d t^\gamma}{\d v_k}.
\end{gather}
Recall that $t^\gamma(v_1,\ldots,v_{N-1})$ is a quasi-homogeneous polynomial of degree $N-\gamma$, if we put $\deg v_a=N-a$. This implies that both sides of the last equation in~\eqref{eq:identity for DN} are zero if $a+b\le N$. Let us assume now that $a+b\ge N+1$. The last equation in~\eqref{eq:identity for DN} is equivalent to 
$$
\frac{\d^2 t^\gamma}{\d v_a\d v_b}-\sum_{i=1}^{a-1}s_{N-i}\frac{\d^2 t^\gamma}{\d v_{a-i}\d v_b}=-\sum_{k=1}^{N-1}\frac{2k-1}{2}\frac{\d t^\gamma}{\d v_k}\left(\omega_{a+b-N-k}-\sum_{i=1}^{a-1}s_{N-i}\omega_{a-i+b-N-k}\right).
$$
Note that for $i\ge a$ we have $a-i+b-N-k<0$ and, therefore, by~\eqref{eq:recursion for omegak}, the expression in the brackets is equal to zero unless $k=a+b-N$. So we come to the following equivalent identity:
\begin{gather}\label{eq:main identity for DN}
\frac{\d^2 t^\gamma}{\d v_a\d v_b}-\sum_{i=1}^{a-1}s_{N-i}\frac{\d^2 t^\gamma}{\d v_{a-i}\d v_b}=-\frac{2(a+b-N)-1}{2}\frac{\d t^\gamma}{\d v_{a+b-N}},\quad 
\begin{array}{@{}c@{}} 1\le a,b,\gamma\le N-1, \\  a+b\ge N+1. \end{array}
\end{gather}

Note that both sides of \eqref{eq:main identity for DN} are quasi-homogeneous polynomials of degree $a+b-\gamma-N$. Differentiating both sides by $\frac{\d^{\sum\alpha_i}}{\d v_1^{\alpha_1}\ldots \d v_{N-1}^{\alpha_{N-1}}}$, putting $v_j=0$ and using formula~\eqref{eq:Dn flat coordinates}, wee see that equation~\eqref{eq:main identity for DN} is equivalent to the following family of identities:
\begin{align}
&-\frac{1}{2}A\left(2\gamma-1+2(N-1)\sum\alpha_i\right)+A\sum_{i=1}^{a-1}(N-i-1)\alpha_{N-i}=-A\frac{2(a+b-N)-1}{2}\Leftrightarrow\notag\\
\Leftrightarrow & A\left(-\gamma-(N-1)\sum\alpha_i+\sum_{i=1}^{a-1}(N-i-1)\alpha_{N-i}+a+b-N\right)=0,\label{eq:final identity for DN}
\end{align} 
that should be true for any tuple $\alpha_1,\ldots,\alpha_{N-1}\ge 0$ such that 
\begin{gather}\label{eq:degree condition for alpha}
\sum_{i=1}^{N-1} (N-i)\alpha_i=a+b-\gamma-N,
\end{gather}
and where $A=\left(-\frac{1}{2}\right)^{\sum\alpha_i} \prod_{k=0}^{\sum\alpha_i-1} \left(2\gamma-1 + 2k(N-1)\right)$. Condition~\eqref{eq:degree condition for alpha} implies that $\alpha_i=0$, if $i\le N-a$. Therefore, the summation $\sum_{i=1}^{a-1}$ in~\eqref{eq:final identity for DN} can be replaced by the summation $\sum_{i=1}^{N-1}$ and, using~\eqref{eq:degree condition for alpha}, we immediately see that the expression in the brackets in~\eqref{eq:final identity for DN} vanishes. This completes the proof of the theorem.
\end{proof}

Taking into account the discussion about the relation between solutions of the open WDVV equations and flat F-manifolds from Section~\ref{subsection:flat F-manifolds and open WDVV}, we get the following result.

\begin{corollary}~\\
1. The function $F^o_{D_N}$ satisfies the open WDVV equations together with condition~\eqref{eq:unit condition for Fo} and the quasi-homogeneity property
$$
\sum_{\alpha=1}^N q_\alpha t^\alpha\frac{\d F^o_{D_N}}{\d t^\alpha}+\frac{1-\delta}{2}s\frac{\d F^o_{D_N}}{\d s}=\frac{3-\delta}{2}F^o_{D_N}.
$$
2. We have
\begin{gather*}
v_k(t^1,\ldots,t^N) = 
\begin{cases}
2^{k-1}(2k-1)\Coef_{s^{2k-1}}F^o_{D_N},&\text{if $1\le k\le N-1$},\\
\sqrt{2\Coef_{s^{-1}}F^o_{D_N}},&\text{if $k=N$}.
\end{cases}
\end{gather*}
\end{corollary}

\begin{example} Here are the Frobenius manifold potentials for the singularities $D_4$ and $D_5$ together with the constructed solutions of the open WDVV equations\footnote{We follow the convention of B.~Dubrovin and use variables with lower indices for the flat coordinates in particular examples.}.
\begin{align*}
F_{D_4} &= \frac{1}{2} t_1^2 t_3 + \frac{1}{2} t_1 t_2^2-\frac{1}{2} t_1 t_4^2 - \frac{1}{4} t_2 t_3 t_4^2 -\frac{1}{12} t_2^3 t_3  +\frac{1}{24} t_2^2 t_3^3 -\frac{1}{24} t_3^3t_4^2 + \frac{t_3^7}{3360},\\
F^o_{D_4} &= \frac{s^7}{168}+\frac{t_3 s^5}{20}  +\left(\frac{t_3^2}{8}+\frac{t_2}{6}\right) s^3+\left(\frac{t_3^3}{12}+\frac{t_2 t_3}{2}+t_1\right) s+\frac{t_4^2}{2 s},\\
F_{D_5} &= \frac{1}{2} t_1^2 t_4 + t_1 t_2 t_3 - \frac{1}{2} t_1 t_5^2  - \frac{1}{6} t_2 t_3^3 - \frac{1}{4} t_2^2 t_3 t_4 + \frac{1}{24} t_2^2 t_4^3 + \frac{1}{6}t_2^3 - \frac{1}{48} t_3^3 t_4^3  + \frac{1}{16} t_3^4 t_4 + \frac{1}{8} t_2 t_3^2 t_4^2+ \\
& \quad + \frac{1}{32256}t_4^9 - \frac{1}{4} t_2 t_4 t_5^2 - \frac{1}{8} t_3 t_4^2 t_5^2 - \frac{1}{64} t_4^4 t_5^2 - \frac{1}{8} t_3^2 t_5^2 + \frac{1}{160} t_3^2 t_4^5,\\
F^o_{D_5} &= \frac{s^9}{576}+\frac{t_4 s^7}{56} +\left(\frac{t_4^2}{16}+\frac{t_3}{20}\right) s^5+\left(\frac{t_4^3}{12}+\frac{t_3 t_4}{4}+\frac{t_2}{6}\right) s^3+\left(\frac{t_4^4}{32}+\frac{t_3 t_4^2}{4} +\frac{t_2 t_4}{2}+\frac{t_3^2}{4}+t_1\right) s+\frac{t_5^2}{2 s}.
\end{align*}
\end{example}


\section{Polynomial solutions of the open WDVV equations for finite irreducible Coxeter groups}\label{sec: polynomial solutions}

In this section we first recall a description of the Frobenius manifolds corresponding to finite irreducible Coxeter groups, and then describe the space of homogeneous polynomial solutions of the associated open WDVV equations.

\subsection{Frobenius manifolds of finite irreducible Coxeter groups}

{\it Finite Coxeter groups} are finite groups of linear transformations of a real $N$-dimensional vector space~$V$, generated by reflections. The complete list of finite irreducible Coxeter groups is given by (the dimension of the space $V$ equals the subscript in the name of the group)
\begin{align}
  & A_N, N\ge 1 && D_N, N \ge 4, && E_6, && E_7, && E_8, &&\label{eq: ADE coxeters}\\
  & B_N, N \ge 2, && F_4, && H_3, && H_4, && I_2(k), k\ge 3, \label{eq: nonADE coxeters}
\end{align}
with the exceptional isomorphisms $A_2\cong I_2(3)$ and $B_2\cong I_2(4)$. By a construction of B.~Dubrovin \cite{Dub98}, for such a group $W$ the complexified space of orbits $M_W := (V \otimes \mbC)/W\cong\mbC^N$ carries a Frobenius manifold structure. For the Coxeter groups $A_N$, $D_N$ and $E_N$ the corresponding Frobenius manifolds coincide with the Saito Frobenius manifolds of simple singularities. By a result of~J.-B. Zuber \cite{Zub94}, the Frobenius manifold potentials corresponding to the remaining irreducible Coxeter groups can be explicitly described by
\begin{align}
F_{B_N}(t^1,\ldots,t^N)=&F_{A_{2N-1}}(t^1,0,t^2,0,\ldots,t^{N-1},0,t^N),\label{eq:formula for non-ADE}\\
F_{I_2(k)}(t^1,t^2)=&F_{A_{k-1}}(t^1,0,\ldots,0,t^2),\notag\\
F_{F_4}(t^1,t^2,t^3,t^4)=&F_{E_6}(t^1,0,t^2,t^3,0,t^4),\notag\\
F_{H_4}(t^1,t^2,t^3,t^4)=&F_{E_8}(t^1,0,t^2,0,0,t^3,0,t^4),\notag\\
F_{H_3}(t^1,t^2,t^3)=&F_{D_6}(t^1,0,t^2,0,t^3,\sqrt{-1}t^2).\notag
\end{align}
All the Frobenius manifolds corresponding to finite irreducible Coxeter groups are semisimple.

\subsection{Euler vector field}

We see that for any finite irreducible Coxeter group $W$, acting on an $N$-dimensional real vector space $V$, the associated Frobenius manifold is described by the polynomial potential $F_W(t^1,\ldots,t^N)$ satisfying the quasi-homogeneity condition 
$$
\sum_{\alpha=1}^N q_\alpha t^\alpha\frac{\d F_W}{\d t^\alpha}=(3-\delta)F_W,\quad q_\alpha>0.
$$
The numbers $q_1,\ldots,q_N$ have the following interpretation. Consider the symmetric algebra~$S(V\otimes\mbC)$. The subring $S(V\otimes\mbC)^W$ of $W$-invariant polynomials is generated by $N$ algebraically independent homogeneous polynomials, whose degrees $d_1,\ldots,d_N\ge 2$ are uniquely determined by the Coxeter group. The maximal degree $h$ is called the {\it Coxeter number} of $W$. Then we have
$$
q_\alpha=\frac{d_\alpha}{h},\qquad \delta=1-\frac{2}{h}.
$$
Note that then in the homogeneity condition~\eqref{eq:homogeneity for Fo} for solutions of the open WDVV equations the degree of the extra variable $s$ becomes
$$
\frac{1-\delta}{2}=\frac{1}{h}.
$$

\subsection{Homogeneous polynomial solutions of the open WDVV equations}

In this section we describe the space of homogeneous polynomial solutions of the open WDVV equations associated to the Frobenius manifolds of finite irreducible Coxeter groups. It occurs that for the Coxeter groups different from $A_N$, $B_N$ and $I_2(k)$ there are no such solutions. We prove it in Section~\ref{subsubsection:nonABI}. For the groups $A_N$, $B_N$ and $I_2(k)$ all solutions can be obtained from the function~$F^o_{A_N}$, as is explained in Section~\ref{subsubsection:ABI}. 

Consider an irreducible Coxeter group $W$, the potential $F_W$ and a homogeneous polynomial solution $F^o$ of the open WDVV equations, satisfying~\eqref{eq:unit condition for Fo}. Note that equations~\eqref{eq:open WDVV,1}-\eqref{eq:unit condition for Fo} involve only the second partial derivatives of $F^o$ and that adding constant and linear terms in the variables $t^1,\ldots,t^N$ and $s$ to $F^o$ just changes the constants $D_\alpha$, $\tD$ and $E$ in condition~\eqref{eq:homogeneity for Fo}. If we remove constant and linear terms in the variables $t^1,\ldots,t^N$ and $s$ from the function $F^o$, then it will satisfy the condition
\begin{gather}\label{eq:homogeneity for open Coxeter}
\sum_{\alpha=1}^N q_\alpha t^\alpha\frac{\d F^o}{\d t^\alpha}+\frac{1-\delta}{2}s\frac{\d F^o}{\d s}=\frac{3-\delta}{2}F^o.
\end{gather}

\subsubsection{Irreducible Coxeter groups different from $A_N$, $B_N$ and $I_2(k)$}\label{subsubsection:nonABI}

\begin{theorem}\label{theorem: coxeters}
Let $W$ be a finite irreducible Coxeter group different from $A_N$, $B_N$ and $I_2(k)$. Consider the corresponding Frobenius manifold potential $F_W$. Then there are no homogeneous polynomial solutions $F^o$ of the associated open WDVV equations satisfying property~\eqref{eq:unit condition for Fo}.
\end{theorem}
\begin{proof}
As we already explained above, we can assume that $F^o$ doesn't contain constant and linear terms in the variables $t^1,\ldots,t^N$ and $s$ and satisfies condition~\eqref{eq:homogeneity for open Coxeter}.

Let $W$ be one of the groups $D_N$, $E_6$, $E_7$ or $E_8$. Let us rewrite equations~\eqref{eq:open WDVV,2} in the coordinates~$v_1,\ldots,v_N$ and $s$:
\begin{gather*}
(c_v)^\mu_{\alpha\beta}\frac{\d^2 F^o}{\d v_\mu\d s}+\frac{\d^2 F^o}{\d v_\alpha\d v_\beta}
\frac{\d^2 F^o}{\d s^2}+\frac{\d t^\talpha}{\d v_\alpha}\frac{\d t^\tbeta}{\d v_\beta}\frac{\d^2 v_\mu}{\d t^\talpha\d t^\tbeta}\frac{\d F^o}{\d v_\mu}\frac{\d^2 F^o}{\d s^2}=\frac{\d^2 F^o}{\d v_\alpha\d s}\frac{\d^2 F^o}{\d v_\beta\d s},\quad 1\le\alpha,\beta\le N,
\end{gather*}
where $(c_v)^\mu_{\alpha\beta}$ denotes the structure constants of multiplication in the coordinates~$v_\mu$. Clearly, $\left.\frac{\d F^o}{\d v_\mu}\right|_{v_*=s=0}=0$. Since $\delta\ge 0$, we have $\frac{3-\delta}{2}>2\cdot\frac{1-\delta}{2}$. This implies that $\left.\frac{\d^2 F^o}{\d s^2}\right|_{v_*=s=0}=0$. Therefore,
\begin{gather*}
\left.\frac{\d}{\d v_\gamma}\left((c_v)^\mu_{\alpha\beta}\frac{\d^2 F^o}{\d v_\mu\d s}\right)\right|_{v_*=s=0}+\left.\frac{\d}{\d v_\gamma}\left(\frac{\d^2 F^o}{\d v_\alpha\d v_\beta}\frac{\d^2 F^o}{\d s^2}\right)\right|_{v_*=s=0}=\left.\frac{\d}{\d v_\gamma}\left(\frac{\d^2 F^o}{\d v_\alpha\d s}\frac{\d^2 F^o}{\d v_\beta\d s}\right)\right|_{v_*=s=0},
\end{gather*}
for any indices $1\le\alpha,\beta,\gamma\le N$. We will prove that this equation can't be true by finding indices $2\le\alpha,\beta,\gamma \le N$ such that
\begin{align}
&\left.(c_v)^\mu_{\alpha\beta}\right|_{v_*=0}=0,&&\frac{\p(c_v)^\mu_{\alpha\beta}}{\p v_\gamma}=A\delta^{\mu,1},\quad A\in\mbC^*,\label{eq:first condition for nonexistence}\\
&\frac{\p^2 F^o}{\p v_\alpha \p v_\beta}=0,&&\frac{\p}{\p v_\gamma} \left(\frac{\p^2 F^o}{\p v_\alpha \p s}\frac{\p^2 F^o}{\p v_\beta \p s}\right) = 0.\label{eq:second condition for nonexistence}
\end{align}

\noindent{\it Case $W = D_N$, $N \ge 4$.} We have $\delta=\frac{N-2}{N-1}$, $\frac{1-\delta}{2}=\frac{1}{2(N-1)}$ and $q_k=\begin{cases}\frac{N-k}{N-1},&\text{if $1\le k\le N-1$},\\\frac{N}{2(N-1)},&\text{if $k=N$}.\end{cases}$ Let us choose $\alpha=2$ and $\beta=\gamma=N$. From $\frac{\p\Lambda_{D_N}}{\p y} = 2xy + v_N$ we see that $\frac{\p}{\p v_2} \circ \frac{\p}{\p v_N} = - \frac{1}{2}v_N \frac{\p}{\p v_1}$, that implies the properties in line~\eqref{eq:first condition for nonexistence}. We have 
\begin{align*}
&q_2 + q_N = \frac{N-2}{N-1} + \frac{N}{2(N-1)} = \frac{3N-4}{2(N-1)} > \frac{2N-1}{2(N-1)} = \frac{3-\delta}{2}&& \hspace{-0.3cm}\Rightarrow &&\hspace{-0.3cm} \frac{\p^2 F^o}{\p v_2 \p v_N} = \frac{\p^3 F^o}{\p v_2 \p v_N\p s} = 0,\\
&2 q_N + \frac{1-\delta}{2} = \frac{2N}{2(N-1)} + \frac{1}{2(N-1)} = \frac{2N+1}{2(N-1)} > \frac{3-\delta}{2} && \hspace{-0.3cm}\Rightarrow &&\hspace{-0.3cm} \frac{\p^3 F^o}{\p v_N^2 \p s} = 0,
\end{align*}
that gives the properties in line~\eqref{eq:second condition for nonexistence}. So the theorem is proved for the case $W=D_N$.

\bigskip

\noindent{\it Case $W = E_6$.} We have $\delta=\frac{5}{6}$, $\frac{1-\delta}{2}=\frac{1}{12}$ and $(q_1,\ldots,q_6)=\left(1,\frac{3}{4},\frac{2}{3},\frac{1}{2},\frac{5}{12},\frac{1}{6}\right)$. Let us choose $\alpha=\beta=\gamma=3$. From $\frac{\p\Lambda_{E_6}}{\p y} = 3y^2 + v_3 + v_5 x + v_6 x^2$ we see that $\frac{\p}{\p v_3} \circ \frac{\p}{\p v_3} = - \frac{1}{3} v_3 \frac{\p}{\p v_1}-\frac{1}{3}v_5\frac{\d}{\d v_2}-\frac{1}{3}v_6\frac{\d}{\d v_4}$, that implies the properties in line~\eqref{eq:first condition for nonexistence}. We have $2q_3 = \frac{4}{3} > \frac{13}{12}=\frac{3-\delta}{2}$, implying $\frac{\p^2 F^o}{\p v_3^2}=\frac{\p^3 F^o}{\p v_3^2\p s} = 0$, that gives the properties in line~\eqref{eq:second condition for nonexistence} and proves the theorem for $W=E_6$.

\bigskip

\noindent{\it Case $W = E_7$.} We have $\delta=\frac{8}{9}$, $\frac{1-\delta}{2}=\frac{1}{18}$ and $(q_1,\ldots,q_7)=\left(1,\frac{7}{9},\frac{2}{3},\frac{5}{9},\frac{4}{9},\frac{1}{3},\frac{1}{9}\right)$. Choose $\alpha=3$, $\beta=4$ and $\gamma=2$. From $\frac{\p \Lambda_{E_7}}{\p x} = 3 x^2y + v_2 + 2 v_4 x + v_5 y + 3 v_6 x^2 + 4 v_7 x^3$ we see that $\frac{\p}{\p v_3} \circ \frac{\p}{\p v_4} = - \frac{1}{3} v_2 \frac{\p}{\p v_1} -\frac{2}{3}v_4\frac{\d}{\d v_2}-\frac{1}{3}v_5\frac{\d}{\d v_3}-v_6\frac{\d}{\d v_4}-\frac{4}{3}v_7\frac{\d}{\d v_6}$, that implies the properties in line~\eqref{eq:first condition for nonexistence}. We have
\begin{align*}
&q_3 + q_4 = \frac{11}{9} > \frac{19}{18} = \frac{3-\delta}{2} && \Rightarrow && \frac{\p^2 F^o}{\p v_3 \p v_4} = 0,\\
&q_2 + q_3 + \frac{1-\delta}{2} = \frac{3}{2} > \frac{3-\delta}{2} && \Rightarrow && \frac{\p^3 F^o}{\p v_2 \p v_3 \p s} = 0,\\
&q_2 + q_4 + \frac{1-\delta}{2} = \frac{25}{18} > \frac{3-\delta}{2} && \Rightarrow && \frac{\p^3 F^o}{\p v_2 \p v_4 \p s} = 0,
\end{align*}
that implies the properties in line~\eqref{eq:second condition for nonexistence}. This proves the theorem for $W=E_7$.
  
\bigskip  
  
\noindent{\it Case $W = E_8$.} We have $\delta=\frac{14}{15}$, $\frac{1-\delta}{2}=\frac{1}{30}$ and $(q_1,\ldots,q_8)=\left(1,\frac{4}{5},\frac{2}{3},\frac{3}{5},\frac{7}{15},\frac{2}{5},\frac{4}{15},\frac{1}{15}\right)$. Choose $\alpha=\beta=\gamma=3$. From $\frac{\p\Lambda_{E_8}}{\p y} = 3 y^2 + v_3 + v_5 x + v_7 x^2 + v_8 x^3$ we see that $\frac{\p}{\p v_3} \circ \frac{\p}{\p v_3} = - \frac{1}{3} v_3 \frac{\p}{\p v_1} -\frac{1}{3}v_5\frac{\d}{\d v_2}-\frac{1}{3}v_7\frac{\d}{\d v_4}-\frac{1}{3}v_8\frac{\d}{\d v_6}$, that implies the properties in line~\eqref{eq:first condition for nonexistence}. We have $2q_3 = \frac{4}{3} > \frac{31}{30} = \frac{3-\delta}{2}$, implying $\frac{\p^2 F^o}{\p v_3^2}=\frac{\p^3 F^o}{\p v_3^2\d s} =0$, that completes the proof of the theorem for $W=E_8$.

\bigskip

For the groups $H_3$, $H_4$ and $F_4$ we are going to use the explicit formulas for the corresponding Frobenius potentials from the paper~\cite{Zub94}. Note that these potentials are related to the ones, given by~\eqref{eq:formula for non-ADE}, by certain rescallings $F_W(t^1,\ldots,t^N)\mapsto F_W(\lambda_1 t^1,\ldots,\lambda_N t^N)$, $\lambda_i\in\mbC^*$, but this doesn't affect our proof.

For the groups $F_4$ and $H_4$ the corresponding potentials, computed in~\cite{Zub94}, are
\begin{align*}
F_{F_4}=&\frac{t_4^{13}}{185328}+\frac{t_3^2 t_4^7}{252}+\frac{t_2^2 t_4^5}{60}+\frac{t_2 t_3^2 t_4^3}{6}+\frac{t_3^4 t_4}{12}+\frac{t_2^3 t_4}{6}+\frac{t_1^2 t_4}{2}+t_1 t_2 t_3,\\
F_{H_4}=&\frac{t_4^{31}}{245764125000}+\frac{t_3^2 t_4^{19}}{1539000}+\frac{t_3^3 t_4^{13}}{10800}+\frac{t_2^2 t_4^{11}}{4950}+\frac{t_2 t_3^2 t_4^9}{360}+\frac{t_3^4 t_4^7}{120} +\frac{t_2^2 t_3 t_4^5}{20}+\frac{t_2 t_3^3 t_4^3}{6}+\frac{t_3^5 t_4}{20}\\
&+\frac{t_2^3 t_4}{6}+\frac{t_1^2 t_4}{2} +t_1 t_2 t_3.
\end{align*}
Note that the equation
\begin{gather*}
\left.\frac{\d}{\d t_2}\left(c^\mu_{2,2}\frac{\d^2 F^o}{\d t_\mu\d s}\right)\right|_{t_*=s=0}+\left.\frac{\d}{\d t_2}\left(\frac{\d^2 F^o}{\d t_2\d t_2}\frac{\d^2 F^o}{\d s^2}\right)\right|_{t_*=s=0}=\left.\frac{\d}{\d t_2}\left(\frac{\d^2 F^o}{\d t_2\d s}\frac{\d^2 F^o}{\d t_2\d s}\right)\right|_{t_*=s=0},
\end{gather*}
where $c^\gamma_{\alpha\beta}$ are the structure constants of multiplication in the coordinates $t_\mu$, can't be true, because $\left.c^\mu_{2,2}\right|_{t_*=0}=0$, $\frac{\p c^\mu_{2,2}}{\p t_2}=\delta^{\mu,1}$ and $\frac{\p^2 F^o}{\p t_2 \p t_2}=0$, that follows from~\eqref{eq:homogeneity for open Coxeter}.

For the group $H_3$ the Frobenius manifold, computed in~\cite{Zub94}, is
\[
    F_{H_3} = \frac{1}{2} t_1^2 t_3 + \frac{1}{2} t_1 t_2^2 +\frac{1}{20} t_2^2 t_3^5 + \frac{1}{6} t_2^3 t_3^2 + \frac{t_3^{11}}{3960}.
\]
The general form of a polynomial function $F^o_{H_3}(t_1,t_2,t_3,s)$ satisfying~\eqref{eq:unit condition for Fo} and~\eqref{eq:homogeneity for open Coxeter} is
\[
    F_{H_3}^o = s t_1  + c_9 s t_2 t_3^2 + c_8 s^3 t_2 t_3  + c_7 s^5 t_2   + c_6 s t_3^5  + c_5 s^3 t_3^4 + c_4 s^5 t_3^3  + c_3 s^7 t_3^2  + c_2 s^9 t_3 + c_1 s^{11},\quad c_k\in\mbC.
\]
Suppose that it satisfies equation~\eqref{eq:open WDVV,2} with $\alpha = 3$, $\beta = 2$. A direct computation shows that, applying the derivative $\frac{\d^2}{\d t_2^2}$ to both sides of it, we get $2$ on the left-hand side and $0$ on the right-hand side. This contradition proves the theorem for the case of the group $H_3$.
\end{proof}

\subsubsection{Coxeter groups $A_N$, $B_N$ and $I_2(k)$}\label{subsubsection:ABI}

Define
\begin{align*}
F^o_{B_N}(t^1,\ldots,t^N,s):=&F^o_{A_{2N-1}}(t^1,0,t^2,0,\ldots,t^{N-1},0,t^N,s),&& N\ge 2,\\
F^o_{I_2(k)}(t^1,t^2,s):=&F^o_{A_{k-1}}(t^1,0,\ldots,0,t^2,s),&& k\ge 3.
\end{align*}
Let $F^{o,-}_{I_2(k)}:=2t^1 s-F^o_{I_2(k)}$ and denote $F^{o,+}_{I_2(k)}:=F^o_{I_2(k)}$. 

Note that if a function $F^o(t^1,\ldots,t^N,s)$ satisfies the open WDVV equations, then the function $\lambda^{-1}F^o(t^1,\ldots,t^N,\lambda s)$ also satisfies them for any $\lambda\ne 0$. Moreover, if $F^o|_{s=0}=0$, then the substitution $\left.\left(\lambda^{-1}F^o(t^1,\ldots,t^N,\lambda s)\right)\right|_{\lambda=0}$ is well defined and is a solution of the open WDVV equations.

\begin{theorem}\label{theorem: An uniqueness}
Let $W$ be one of the groups $A_N$, $B_N$ or $I_2(k)$. Then all solutions $F^o$ of the open WDVV equations satisfying~\eqref{eq:unit condition for Fo} and~\eqref{eq:homogeneity for open Coxeter} are given by the family 
$$
F^o=
\left\{
\begin{aligned}
&\lambda^{-1}F^o_{A_N}(t^1,\ldots,t^N,\lambda s),&&\lambda\in\mbC^*,&&\text{if $W=A_N$, $N\ge 2$},\\
&\lambda^{-1}F^o_{A_1}(t^1,\lambda s),&&\lambda\in\mbC,&&\text{if $W=A_1$},\\
&\lambda^{-1}F^o_{B_N}(t^1,\ldots,t^N,\lambda s),&&\lambda\in\mbC,&&\text{if $W=B_N$, $N\ge 2$},\\
&\lambda^{-1}F^o_{I_2(k)}(t^1,t^2,\lambda s),&&\lambda\in\mbC^*,&&\text{if $W=I_2(k)$, $k$ is odd},\\
&\lambda^{-1}F^{o,\pm}_{I_2(k)}(t^1,t^2,\lambda s),&&\lambda\in\mbC,&&\text{if $W=I_2(k)$, $k$ is even}.
\end{aligned}\right.
$$
\end{theorem}
\begin{proof}
{\it Case $W=A_N$}. We have $q_\alpha=\frac{N+2-\alpha}{N+1}$ and $\delta=\frac{N-1}{N+1}$. The case $N=1$ is obvious. Suppose that $N\ge 2$ and $F^o$ is a solution of the open WDVV equations, satisfying~\eqref{eq:unit condition for Fo} and~\eqref{eq:homogeneity for open Coxeter}. For an $n$-tuple $\oalpha=(\alpha_1,\ldots,\alpha_n)$, $1\le\alpha_i\le N$, denote
$$
\<\tau_\oalpha\sigma^k\>^o=\<\tau_{\alpha_1}\ldots\tau_{\alpha_n}\sigma^k\>^o:=\left.\frac{\d^{n+k}F^o}{\d t^\alpha_1\ldots\d t^{\alpha_n}\d s^k}\right|_{t^*=s=0}.
$$ 
This number is non-zero only if $k=k(\oalpha):=N+2-\sum_{i=1}^n(N+2-\alpha_i)$. 

Note that 
\begin{gather}\label{eq:c2betagamma}
c^\gamma_{2,\beta}=
\begin{cases}
0,&\text{if $\gamma>\beta+1$},\\
1,&\text{if $\gamma=\beta+1$},\\
O(t^*),&\text{if $\gamma\le\beta$},
\end{cases}
\end{gather}
that follows from~\eqref{eq:3-point and 4-point for AN}. Setting $t^*=0$ in equation~\eqref{eq:open WDVV,2} with $\alpha=2$ and $2\le\beta\le N$, we get
\begin{gather*}
\<\tau_\alpha\sigma^\alpha\>^o=(\alpha-1)!\left(\<\tau_2\sigma^2\>^o\right)^{\alpha-1},\quad 2\le\alpha\le N,\qquad\<\tau_2\tau_N\>^o\<\sigma^{N+2}\>^o=N!\left(\<\tau_2\sigma^2\>^o\right)^N.
\end{gather*}
Differentiating equation~\eqref{eq:open WDVV,2} with $\alpha=2$ and $\beta=N$ by~$\frac{\d}{\d t^2}$ and setting $t^*=s=0$, we get $-1+\<\tau_2\tau_N\>^o\<\tau_2\sigma^2\>^o=0$, where we use formula~\eqref{eq:3-point and 4-point for AN} for the numbers $\<\tau_{\alpha_1}\tau_{\alpha_2}\tau_{\alpha_3}\tau_{\alpha_4}\>_{A_N}$. We see that $\<\tau_2\sigma^2\>^o\ne 0$ and
$$
\<\sigma^{N+2}\>^o=N!\left(\<\tau_2\sigma^2\>^o\right)^{N+1}\ne 0.
$$

After the rescaling $F^o(t^1,\ldots,t^N,s)\mapsto\lambda^{-1}F^o(t^1,\ldots,t^N,\lambda s)$ with an appropriate constant $\lambda\ne 0$ we get $\<\tau_2\sigma^2\>^o=1$ and, therefore,
\begin{gather*}
\<\tau_\alpha\sigma^\alpha\>^o=(\alpha-1)!=\<\tau_\alpha\sigma^\alpha\>^o_{A_N},\quad 1\le\alpha\le N,\qquad\<\sigma^{N+2}\>^o=N!=\<\sigma^{N+2}\>^o_{A_N}.
\end{gather*}

Consider now an $n$-tuple $\oalpha=(\alpha_1,\ldots,\alpha_n)$, $1\le\alpha_i\le N$, with $n\ge 2$ and $k(\oalpha)\ge 0$. Differentiating equation~\eqref{eq:open WDVV,2} with $\alpha=\alpha_1$ and $\beta=\alpha_2$ by $\frac{\d^{n-2}}{\d t^{\alpha_3}\ldots\d t^{\alpha_n}}$ and setting $t^*=0$, we get the recursion
\begin{gather}\label{eq:recursion for AN}
\frac{\<\tau_{\oalpha}\sigma^{k(\oalpha)}\>^o}{k(\oalpha)!}=\sum_{\substack{I\sqcup J=\{1,\ldots,n\}\\1\in I,\,2\in J}}\frac{\<\tau_{\oalpha_I}\sigma^{k(\oalpha_I)}\>^o\<\tau_{\oalpha_J}\sigma^{k(\oalpha_J)}\>^o}{(k(\oalpha_I)-1)!(k(\oalpha_J)-1)!}-\sum_{\substack{I\sqcup J=\{1,\ldots,n\}\\1,2\in I,\,J\ne\emptyset}}\frac{\<\tau_{\oalpha_I}\sigma^{k(\oalpha_I)}\>^o\<\tau_{\oalpha_J}\sigma^{k(\oalpha_J)}\>^o}{k(\oalpha_I)!(k(\oalpha_J)-2)!},
\end{gather}
where for a subset $I=\{i_1,\ldots,i_{|I|}\}\subset\{1,\ldots,n\}$, $i_1<\ldots<i_{|I|}$, we denote $\oalpha_I:=(\alpha_{i_1},\ldots,\alpha_{i_{|I|}})$. The correlators $\<\cdot\>_{A_N}$ don't appear in this recursion because for any subset $I\subset\{1,\ldots,n\}$ and an index $1\le\mu\le N$ we have
\begin{multline*}
\sum_{i\in I}(N+2-\alpha_i)+(N+2-\mu)\le\sum_{i=1}^n(N+2-\alpha_i)+(N+2-\mu)=\\
=2N+4-k(\oalpha)-\mu<2N+4
\end{multline*}
and, therefore, $\<\tau_{\oalpha_I}\tau_\mu\>_{A_N}=0$. The recursion~\eqref{eq:recursion for AN} determines all the numbers $\<\tau_{\oalpha}\sigma^{k(\oalpha)}\>^o$ starting from the numbers~$\<\sigma^{N+2}\>^o$ and $\<\tau_\alpha\sigma^\alpha\>^o$. So we conclude that $F^o=F^o_{A_N}$.

\bigskip

\noindent{\it Case $W=B_N$}. We have $q_\alpha=\frac{N+1-\alpha}{N}$ and $\delta=\frac{N-1}{N}$. The function $F^o_{B_N}$ satisfies the open WDVV equations together with equations~\eqref{eq:unit condition for Fo} and \eqref{eq:homogeneity for open Coxeter}, because, as one can easily check using the quasi-homogeneity of the function $F^o_{A_{2N-1}}$, the correlator $\<\tau_{\alpha_1}\ldots\tau_{\alpha_n}\tau_\mu\>_{A_{2N-1}}$ vanishes, if all the~$\alpha_i$'s are odd and $\mu$ is even.

Suppose that $F^o$ is a solution of the open WDVV equations, satisfying~\eqref{eq:unit condition for Fo} and~\eqref{eq:homogeneity for open Coxeter}. Since $F_{B_2}=F_{I_2(4)}$, we will consider the $B_2$-case together with the cases $W=I_2(k)$ later. So we assume that $N\ge 3$. Note that a correlator $\<\tau_{\alpha_1}\ldots\tau_{\alpha_n}\sigma^k\>^o$ vanishes unless $\sum_{i=1}^N(N+1-\alpha_i)+\frac{k}{2}=N+\frac{1}{2}$. Setting $t^*=0$ in equation~\eqref{eq:open WDVV,2} with $\alpha=2$, we get the relations
\begin{align*}
\<\tau_\alpha\sigma^{2\alpha-1}\>^o=&\frac{(2\alpha-2)!}{2^{\alpha-1}}\left(\<\tau_2\sigma^3\>^o\right)^{\alpha-1},\quad 2\le\alpha\le N,\\
\<\tau_2\tau_N\sigma\>^o\<\sigma^{2N+1}\>^o=&\frac{(2N-1)!}{2^N}\left(\<\tau_2\sigma^3\>^o\right)^N.
\end{align*}
Differentiating equation~\eqref{eq:open WDVV,2} with $\alpha=2$ and $\beta=N-1$ by $\frac{\d}{\d t^2}$ and setting $t^*=s=0$, we get $\<\tau_2^2\tau_{N-1}\tau_N\>_{B_N}+\<\tau_2\tau_N\sigma\>^o=0$. Since, by~\eqref{eq:3-point and 4-point for AN}, $\<\tau_2^2\tau_{N-1}\tau_N\>_{B_N}=-1$, we conclude that $\<\tau_2\tau_N\sigma\>^o=1$ and
$$
\<\sigma^{2N+1}\>^o=\frac{(2N-1)!}{2^N}\left(\<\tau_2\sigma^3\>^o\right)^N.
$$

Suppose that $\<\tau_2\sigma^3\>^o\ne 0$, then $\<\sigma^{2N+1}\>^o\ne 0$. After the rescaling $F^o(t^1,\ldots,t^N,s)\mapsto\lambda^{-1}F^o(t^1,\ldots,t^N,\lambda s)$ with an appropriate constant $\lambda\ne 0$ we get $\<\tau_\alpha\sigma^{2\alpha-1}\>^o=\<\tau_\alpha\sigma^{2\alpha-1}\>^o_{B_N}$ and $\<\sigma^{2N+1}\>^o=\<\sigma^{2N+1}\>^o_{B_N}$. In the same way, as in the $A_N$-case, there is a recursion similar to~\eqref{eq:recursion for AN}, that reconstructs all the correlators $\<\tau_{\alpha_1}\ldots\tau_{\alpha_n}\sigma^k\>^o$ with $n\ge 2$. Therefore, $F^o=F^o_{B_N}$.

Suppose that $\<\tau_2\sigma^3\>^o=0$, then $\<\sigma^{2N+1}\>^o=0$. Consider the decomposition
$$
F^o=\sum_{k=0}^N P_k(t^1,\ldots,t^N)s^{2k+1},\quad P_k\in\mbC[t^1,\ldots,t^N].
$$
Consider an index $l$ such that $P_l\ne 0$ and $P_{>l}=0$. Since $l<N$, the polynomial~$P_l$ can't be a constant. Suppose $l>0$, then equation~\eqref{eq:open WDVV,2} implies that
$$
\frac{\d P_l}{\d t^\alpha}\frac{\d P_l}{\d t^\beta}=\frac{2l}{2l+1}P_l\frac{\d^2 P_l}{\d t^\alpha\d t^\beta},\quad 1\le\alpha,\beta\le N.
$$
The space of solutions of the differential equation $(f')^2=\frac{2l}{2l+1}ff''$ for a function $f=f(x)$ is formed by the family $f=C_1(x+C_2)^{-2l}$, $C_1,C_2\in\mbC^*$, together with the constant solution $f=C$, $C\in\mbC$. Since $P_l$ is a non-constant polynomial, we come to a contradition. Therefore, $l=0$. 

In this case system~\eqref{eq:open WDVV,2} is equivalent to the system
$$
c^\gamma_{\alpha\beta}\frac{\d P_0}{\d t^\gamma}=\frac{\d P_0}{\d t^\alpha}\frac{\d P_0}{\d t^\beta},\quad 1\le\alpha,\beta\le N.
$$
For $\alpha=2$ we get the relations
$$
\frac{\d P_0}{\d t^{\beta+1}}+\sum_{1\le\gamma\le\beta}c^\gamma_{2,\beta}\frac{\d P_0}{\d t^\gamma}=\frac{\d P_0}{\d t^2}\frac{\d P_0}{\d t^\beta},\quad 2\le\beta\le N-1,
$$
that recursively determine all the derivatives $\frac{\d P_0}{\d t^\beta}$ starting from the derivatives $\frac{\d P_0}{\d t^2}=t^N$ and $\frac{\d P_0}{\d t^1}=1$. This completely determines the polynomial $P_0$. We conclude that $F^o=\left.\left(\lambda^{-1}F^o_{B_N}(t^1,\ldots,t^N,\lambda s)\right)\right|_{\lambda=0}$. 

\bigskip

\noindent{\it Case $W=I_2(k)$}. We have $q_1=1$, $q_2=\frac{2}{k}$, $\delta=\frac{k-2}{k}$ and $F_{I_2(k)}=\frac{(t^1)^2t^2}{2}+\alpha_k\frac{(t^2)^{k+1}}{(k+1)!}$, $\alpha_k\ne 0$. The function $F^o_{I_2(k)}$ satisfies the open WDVV equations together with equations~\eqref{eq:unit condition for Fo} and \eqref{eq:homogeneity for open Coxeter}, because $\<\tau_{\alpha_1}\ldots\tau_{\alpha_n}\tau_\mu\>_{A_{k-1}}=0$, if $\alpha_i\in\{1,k-1\}$ and $\mu\notin\{1,k-1\}$~\cite[Section~1]{Zub94}.

Note that if a function $F^o$ satisfies property~\eqref{eq:unit condition for Fo}, then all the open WDVV equations are automatically satisfied except equation~\eqref{eq:open WDVV,2} with $\alpha=\beta=2$.

Suppose $k=2l+1$, $l\ge 1$. A polynomial $F^o(t^1,t^2,s)$, satisfying~\eqref{eq:unit condition for Fo} and~\eqref{eq:homogeneity for open Coxeter}, has the form
$$
F^o=t^1 s+\sum_{i=0}^{l+1}\beta_i\frac{s^{2l+2-2i}}{(2l+2-2i)!}\frac{(t^2)^i}{i!},\quad\beta_i\in\mbC.
$$
Suppose that the open WDVV equations are satisfied. Equation~\eqref{eq:open WDVV,2} with $\alpha=\beta=2$ is equivalent to
\begin{align*}
&\frac{\d^2 F_{I_2(2l+1)}}{\d(t^2)^3}+\frac{\d^2 F^o}{\d(t^2)^2}\frac{\d^2 F^o}{\d s^2}-\left(\frac{\d^2 F^o}{\d t^2\d s}\right)^2=0\Leftrightarrow\\
\Leftrightarrow & \alpha_{2l+1}\frac{(t^2)^{2l-1}}{(2l-1)!}+\left(\sum_{i=2}^{l+1}\beta_i\frac{s^{2l+2-2i}}{(2l+2-2i)!}\frac{(t^2)^{i-2}}{(i-2)!}\right)\left(\sum_{i=0}^l\beta_i\frac{s^{2l-2i}}{(2l-2i)!}\frac{(t^2)^i}{i!}\right)\\
&-\left(\sum_{i=1}^l\beta_i\frac{s^{2l+1-2i}}{(2l+1-2i)!}\frac{(t^2)^{i-1}}{(i-1)!}\right)^2=0.
\end{align*}
The expression on the left-hand side of the last equation has the form $\sum_{i=0}^{2l-1}(t^2)^{2l-1-i}s^{2i}P_i(\beta_0,\ldots,\beta_{l+1})$, where
\begin{gather*}
P_0=\frac{\alpha_{2l+1}}{(2l-1)!}+\frac{\beta_{l+1}\beta_l}{(l-1)!l!},\qquad P_i=\frac{\beta_{l+1}\beta_{l-i}}{(l-1)!(l-i)!(2i)!}+Q_i(\beta_{l-i+1},\ldots,\beta_{l+1}),\quad 1\le i\le l,
\end{gather*}
and $Q_i$ are polynomials in $\beta_{l-i+1},\ldots,\beta_{l+1}$. We see that $\beta_{l+1}\ne 0$ and the equations $P_i=0$, $0\le i\le l$, determine the coefficients $\beta_0,\ldots,\beta_l$ in terms of the coefficient $\beta_{l+1}$. Thus, $F^o=\lambda^{-1}F^o_{I_2(2l+1)}(t^1,t^2,\lambda s)$ for some $\lambda\ne 0$.

Suppose $k=2l$, $l\ge 2$, and a polynomial $F^o$ satisfies the open WDVV equations together with equations~\eqref{eq:unit condition for Fo} and~\eqref{eq:homogeneity for open Coxeter}. Then $F^o$ has the form
$$
F^o=t^1 s+\sum_{i=0}^l\beta_i\frac{s^{2l+1-2i}}{(2l+1-2i)!}\frac{(t^2)^i}{i!},\quad\beta_i\in\mbC,
$$
and equation~\eqref{eq:open WDVV,2} with $\alpha=\beta=2$ is equivalent to
\begin{align*}
& \alpha_{2l}\frac{(t^2)^{2l-2}}{(2l-2)!}+\left(\sum_{i=2}^l\beta_i\frac{s^{2l+1-2i}}{(2l+1-2i)!}\frac{(t^2)^{i-2}}{(i-2)!}\right)\left(\sum_{i=0}^{l-1}\beta_i\frac{s^{2l-1-2i}}{(2l-1-2i)!}\frac{(t^2)^i}{i!}\right)\\
&-\left(\sum_{i=1}^l\beta_i\frac{s^{2l-2i}}{(2l-2i)!}\frac{(t^2)^{i-1}}{(i-1)!}\right)^2=0.
\end{align*}
The expression on the left-hand side has the form $\sum_{i=0}^{2l-2}(t^2)^{2l-2-i}s^{2i}P_i(\beta_0,\ldots,\beta_l)$, where
\begin{align*}
P_0=&\frac{\alpha_{2l}}{(2l-2)!}-\frac{\beta_l^2}{((l-1)!)^2},\\
P_i=&\gamma_i\beta_l\beta_{l-i}+Q_i(\beta_{l-i+1},\ldots,\beta_l),\quad \gamma_i=\frac{2l(i-1)}{(l-1)!(2i)!(l-i)!},\quad 1\le i\le l,
\end{align*}
and $Q_i$ are polynomials in $\beta_{l-i+1},\ldots,\beta_l$. We see that the equation $P_0=0$ determines $\beta_l$ up to a sign and then the equations $P_i=0$, $2\le i\le l$, determine the coefficients $\beta_0,\ldots,\beta_{l-2}$ in terms of $\beta_{l-1}$. Thus, $F^o=\lambda^{-1}F^{o,\pm}_{I_2(2l)}(t^1,t^2,\lambda s)$ for some $\lambda$.
\end{proof}


\begin{thebibliography}{BCT18}

\bibitem[Bur18]{Bur18} A. Buryak. {\it Extended $r$-spin theory and the mirror symmetry for the $A_{r-1}$-singularity}. arXiv:1802.07075v2.

\bibitem[BCT18]{BCT18} A. Buryak, E. Clader, R. J. Tessler. {\it Open $r$-spin theory and the Gelfand--Dickey wave function}. arXiv:1809.02536v3.

\bibitem[BCT19]{BCT19} A. Buryak, E. Clader, R. J. Tessler. {\it Closed extended $r$-spin theory and the Gelfand--Dickey wave function}. Journal of Geometry and Physics 137 (2019), 132--153.

\bibitem[Dub96]{Dub96} B. Dubrovin. {\it Geometry of 2D topological field theories}. Integrable systems and quantum groups (Montecatini Terme, 1993), 120--348, Lecture Notes in Math., 1620, Fond. CIME/CIME Found. Subser., Springer, Berlin, 1996. 

\bibitem[Dub98]{Dub98} B. Dubrovin. {\it Differential geometry of the space of orbits of a Coxeter group}. Surveys in Differential Geometry~4 (1998), 181--211.

\bibitem[Dub99]{Dub99} B. Dubrovin. {\it Painlev\'e transcendents in two-dimensional topological field theory}. The Painlev\'e property, 287--412, CRM Ser. Math. Phys., Springer, New York, 1999.

\bibitem[FJR13]{FJR13} H. Fan, T. Jarvis, Y. Ruan. {\it The Witten equation, mirror symmetry, and quantum singularity theory}. Annals of Mathematics~178 (2013), 1--106.

\bibitem[Hert02]{Hert02} C. Hertling. {\it Frobenius manifolds and moduli spaces for singularities}. Cambridge Tracts in Mathematics, 151. Cambridge University Press, Cambridge, 2002.

\bibitem[HS12]{HS12} A. Horev, J. P. Solomon. {\it The open Gromov--Witten--Welschinger theory of blowups of the projective plane}. arXiv:1210.4034v1.

\bibitem[JKV01a]{JKV01a} T. J. Jarvis, T. Kimura, A. Vaintrob. {\it Moduli spaces of higher spin curves and integrable hierarchies}. Compositio Mathematica 126 (2001), no. 2, 157--212.

\bibitem[JKV01b]{JKV01b} T. J. Jarvis, T. Kimura, A. Vaintrob. {\it Gravitational descendants and the moduli space of higher spin curves}. Advances in algebraic geometry motivated by physics (Lowell, MA, 2000), 167--177, Contemp. Math., 276, Amer. Math. Soc., Providence, RI, 2001. 

\bibitem[LVX17]{LVX17} K. Liu, R. Vakil, H. Xu. {\it Formal pseudodifferential operators and
Witten’s $r$-spin numbers}. Journal f\"ur die reine und angewandte Mathematik 728 (2017), 1--33.

\bibitem[Man99]{Man99} Yu. I. Manin. {\it Frobenius manifolds, quantum cohomology, and moduli spaces}. American Mathematical Society Colloquium Publications, 47. American Mathematical Society, Providence, RI, 1999.

\bibitem[Man05]{Man05} Yu. I. Manin. {\it F-manifolds with flat structure and Dubrovin's duality}. Advances in Mathematics 198 (2005), no. 1, 5--26.

\bibitem[NY98]{NY98} M. Noumi, Y. Yamada. {\it Notes on the flat structures associated with simple and simply elliptic singularities}. Integrable systems and algebraic geometry (Kobe/Kyoto, 1997), 373--383, World Sci. Publ., River Edge, NJ, 1998.

\bibitem[PPZ16]{PPZ16} R. Pandharipande, A. Pixton, D. Zvonkine. {\it Tautological relations via $r$-spin structures}. arXiv:1607.00978v2.

\bibitem[PST14]{PST14} R. Pandharipande, J. P. Solomon, R. J. Tessler. {\it Intersection theory on moduli of disks, open KdV and Virasoro}. arXiv:1409.2191v2.

\bibitem[Sai82]{Sai82} K. Saito. {\it Primitive forms for a universal unfolding of a function with an isolated critical point}. Journal of the Faculty of Science, University of Tokyo, Section IA, Mathematics~28 (1981), no.~3, 775--792 (1982).

\bibitem[Sai83]{Sai83} K. Saito. {\it Period mapping associated to a primitive form}. Kyoto University, Research Institute for Mathematical Sciences, Publications~19 (1983), no.~3, 1231--1264.

\bibitem[Zub94]{Zub94} J.-B. Zuber. {\it On Dubrovin topological field theories}. Modern Physics Letters A~9 (1994), no.~8, 749--760.

\end{thebibliography}
\end{document}